\documentclass[journal]{IEEEtran}
\usepackage{citesort}
\usepackage[dvips]{graphicx}
\usepackage[cmex10]{amsmath}
\usepackage{algorithm}
\usepackage{algpseudocode}
\usepackage{array}
\usepackage[tight,footnotesize]{subfigure}
\usepackage{setspace}
\usepackage{enumerate} 
\usepackage{url}

\usepackage{amsfonts, amssymb, amsthm, mathrsfs}
\usepackage{setspace}
\usepackage{Tabbing}
\usepackage{fancyhdr}
\usepackage{lastpage}
\usepackage{extramarks}
\usepackage{chngpage}
\usepackage{soul,color}
\usepackage{indentfirst}
\usepackage{float,wrapfig}
\usepackage{times}
\usepackage{multirow}
\usepackage{bm}
\usepackage{bbm}

\usepackage{lipsum}

\DeclareMathOperator*{\argmax}{arg\,max}

\newtheorem{theo}{Theorem}
\newtheorem{defi}{Definition}

\newtheorem{prop}{Proposition}

\newtheorem{cor}{Corollary}

\newcolumntype{I}{!{\vrule width 1.5pt}}
\newlength\savedwidth

\newcommand{\sinr}{\textrm{SINR}} 
   
\newcommand{\rt}{\textrm{rate}}     

\newcommand{\tphi}{\tilde{\Phi}}

\newcommand{\tabincell}[2]{\begin{tabular}{@{}#1@{}}#2\end{tabular}}               
                  
\begin{document}

\title{Resource Optimization in Device-to-Device Cellular Systems Using Time-Frequency Hopping}

\author{Qiaoyang Ye, 
        Mazin Al-Shalash,
        Constantine Caramanis
        and Jeffrey G. Andrews
\thanks{Q. Ye, C. Caramanis and  J. G. Andrews are with WNCG, The University of Texas at Austin, USA, M. Shalash is with Huawei Technologies. Email: 
qye@utexas.edu, mshalash@huawei.com, constantine@utexas.edu, jandrews@ece.utexas.edu. A part of this
paper will be presented at IEEE ICC 2014 \cite{YeICC14}.
Manuscript last revised: \today.}}

\maketitle

\begin{abstract}

We develop a flexible and accurate framework for device-to-device (D2D) communication in the context of a  conventional cellular network, which allows for time-frequency resources to be either shared or orthogonally partitioned between the two networks. 
Using stochastic geometry, we provide accurate expressions for SINR distributions and average rates, under an assumption of interference randomization via time and/or frequency hopping, for both dedicated and shared spectrum approaches. We obtain analytical results in closed or semi-closed form in high SNR regime, that allow us to easily explore the impact of key parameters (e.g., the load and hopping probabilities) on the network performance. In particular, unlike other models, the expressions we obtain are tractable, i.e., they can be efficiently optimized without extensive simulation. Using these, we  optimize the hopping probabilities for the D2D links, i.e., how often they should request a time or frequency slot. This can be viewed as an optimized lower bound to other more sophisticated scheduling schemes. We also investigate the optimal resource partitions between D2D and cellular networks when they use orthogonal resources.
\end{abstract}


\IEEEpeerreviewmaketitle

\section{Introduction}
By allowing direct communication between physically proximal devices, device-to-device (D2D) communication can reduce energy consumption, efficiently utilize the network resources, reduce end-to-end latency, and increase the network capacity and flexibility. Consequently, D2D communication is emerging as a potentially important  technology component for LTE-Advanced, aiming to meet the growing demand for local wireless services~\cite{DopRin09,CorLar10,WuTav10,FodDah12}. Unlike general ad hoc networks, D2D can benefit from cellular infrastructure (e.g., network coordinated device discovery, synchronization and enhanced security), and can operate on licensed bands, which makes resource allocation more tractable and reliable. 

In the D2D-enabled cellular network, D2D links can either use orthogonal resources or share resources with the cellular network. In a network with orthogonal allocation -- called a \textit{dedicated network} -- the interference management is simplified, but the resource utilization may be less efficient. On the other hand, if D2D transmissions reuse cellular resources -- called a \textit{shared network} -- network resources can be used more efficiently, at the cost of a denser interference environment, which complicates interference management.  Which is preferable?  Potential D2D data can either be transmitted directly (D2D), or via a base station (BS) -- termed \emph{mode selection}.   When should a potential D2D link transmit directly, versus relaying via the BS? As we explain in detail below, we develop a flexible model  to answer these questions by providing accurate analytical results and simple semi-closed form expressions for performance bounds, which in turn are amenable to efficient optimization.

\vspace{-0.4cm}
\subsection{Related Work}
Paper \cite{HuaLau09} investigates both the dedicated and shared approaches for uplink resources and shows that in general, the dedicated approach is more efficient in terms of \emph{transmission capacity} -- i.e., it allows more successful transmissions per unit area. On the other hand, in terms of total rate, \cite{YuDop11} shows the shared approach is better  in a single cell scenario  with a maximal rate cap, taking into account both the uplink and downlink transmissions. For a shared network, careful resource allocation can control the mutual interference between cellular  and D2D transmissions. For example, an intelligent frequency allocation where orthogonal resources are assigned to nearby cellular and D2D links~\cite{XuWan10}, exclusive D2D transmission zones \cite{MinLee11}, mixed integer nonlinear programming problems \cite{ZulHua10,PhuHos13}, auction based mechanisms in a downlink single cell \cite{XuSon12,XuSon13}, a Stackelberg game framework in an uplink single cell \cite{WanSon13}, and  interference randomization through time hopping~\cite{CheCha10} are viable approaches for controlling interference. Besides, performance analysis considering interference among D2D links is conducted in \cite{LeiZha13}.  
For mode selection, simple distance-based and received signal-based mode selections are proposed in~\cite{XiaPen12} and  \cite{AdaNak98}, respectively. More sophisticated mode selection involving other user equipments (UEs) are proposed in~\cite{HakChe10,DopYu10,JunHwa12}. Nevertheless, the majority of earlier studies consider a single cell scenario and propose heuristic algorithms to improve network performance. In this paper, we leverage tools from stochastic geometry to study a more general D2D-enabled cellular network.

In a pure ad hoc network, there has been significant success over the past decade in proposing tractable models for performance analysis and system design via stochastic geometry~\cite{HaeAnd09}. For example,~\cite{BacBla06} investigates an Aloha-type access mechanism for a large ad hoc network, while~\cite{BacLi11} analyzes a carrier sense multiple access (CSMA)-type mechanism. The outage probability and transmission capacity of ad hoc networks are studied and summarized in~\cite{WebYan05,WebAnd10}.  There have been some analogous more recent results for cellular networks \cite{AndBac10,DhiGan12}, where the BSs are modeled as a Poisson point process (PPP).  
D2D-enabled cellular networks are essentially a combination of cellular and ad hoc networks, but combining these models into a more general framework is nontrivial.  D2D communication can either utilize uplink or downlink resources, and it is not {\it a priori} clear which resource utilization is better. There has been at least one very recent (parallel) work attempting this for the uplink system  \cite{lin2013uplinkD2D}.  For comparison and the completeness of study, we instead investigate a D2D-enabled  cellular network, where downlink resources are either partitioned or shared between D2D and downlink cellular transmissions.

\vspace{-0.4cm}
\subsection{Contributions}
\vspace{-0.1cm}
The objective of this paper is to propose a general framework for the analysis of system performance (e.g., the signal-to-interference-plus-noise ratio (SINR) distribution and total rate) in D2D-enabled cellular networks. We apply this framework to both dedicated and shared downlink networks, which are easy to analyze and optimize, and can be adopted as flexible baseline models for further study.   Our key contributions are enabled by simultaneously leveraging techniques from stochastic geometry and optimization~theory.

\textbf{Tractable model for both dedicated and shared cellular networks.} We propose a tractable model for a large D2D-enabled cellular network, where the locations of BSs and UEs are modeled as spatial point processes, particularly PPPs. We propose to adopt a time-frequency hopping scheme for potential D2D links to randomize the interference, where each potential D2D link chooses its operation mode (i.e., D2D or cellular mode) at each time slot independently according to a predefined time hopping probability, and accesses each subband independently with a predefined frequency hopping probability. In this model, the derived SINR distributions have remarkably simple forms, which provide an efficient system performance evaluation without time-consuming simulations.  It is not always possible to get the expected rate in closed form. We provide a general expression for the average rate and then derive its lower bound, which is in a semi-closed form in interference-limited networks. 

\textbf{ Optimization of network performance and design insights.} Based on the derived SINR distributions and the lower bounds on  average rates, we investigate the optimal D2D hopping probabilities (i.e., how often potential D2D links should request a time or frequency slot) using optimization theory. The optimal network performance can serve as a lower bound for D2D-enable cellular networks with more sophisticated scheduling scheme.
We find that in many cases, we can either derive the optimal solution in a simple closed-form, or reduce the problem to lower dimension (e.g., one of the hopping probabilities is found in closed-form).   The observed design principles are now summarized.

\textit{Dedicated vs. shared.}  Unsurprisingly, the dedicated network has better SINRs since resources allocated to D2D and cellular links are orthogonal. With an optimal spectrum partition between D2D and cellular users, the dedicated network also provides larger average rate, but should be interpreted cautiously.  For example, the optimal spectrum partition may be very hard to determine, or it may vary significantly in time or space over a non-homogeneous network (recall we model all BSs and UEs as  homogenous PPPs).  In such cases, the shared approach may be able to perform significantly better, as it is more flexible. For cases with small amount of local traffic, the shared approach may also have a better performance.

\emph{Optimal hopping scheme.}  In the dedicated network with any general non-decreasing utility function, the optimal D2D frequency hopping depends on the service demands of D2D users (i.e., the traffic arrival rate). D2D links with more traffic to transmit should be more aggressive in their spectrum access, despite the interference that this generates to the rest of the network. The same observation can be done from simulation results of the shared network.

As for time hopping, in most considered interference limited cases with heavy load, all potential D2D links should operate in D2D mode (bypassing the BS), assuming the objective is to maximize the total average rate. This result is independent of the average distance between a D2D transmitter and its receiver, which is perhaps surprising, and largely due to the use of total average rate as the utility function. We demonstrate this by giving an example in Section V-A, showing that the optimal mode selection for different utility functions may be very different.  In principle, any utility function can be investigated based on the proposed framework, but we use total average rate in a heavily loaded network, and leave other utility functions to future work.


The paper is organized as follows. We present the system model in Section \ref{sec:model}. In Sections \ref{sec:overlay} and~\ref{sec:underlay}, we analyze the SINR distributions and average rates in dedicated and shared networks, respectively. We investigate the optimization in terms of   hopping probabilities in Section~\ref{sec:opt}. The numerical results are given in Section \ref{sec:simulation}. Finally, we conclude and suggest possible extensions in Section~\ref{sec:conclusion}.


\section{System Model}\label{sec:model}
We focus on a downlink model, where D2D communication uses downlink cellular resources. The key aspects of the model are described in the following subsections. 
\vspace{-0.5cm}
\subsection{Deployment of D2D and cellular networks}
We consider a large D2D-enabled downlink cellular network, illustrated in Fig.~\ref{fig:network}. We classify the potential D2D transmitters into $M$ types which may differ in terms of their service demands and/or the MAC protocol. 
Note that similar to the current wireless traffic growth driven by the smartphones proliferating around the world, more local traffic will possibly be generated once the D2D features are available in future. Therefore, at this stage, the D2D traffic demand as well as its growth is not clear. Though any general distributions can be used to model the location of D2D users, random (uniform) dropping is one of the most popular models in both academia and industry (e.g., \cite{ALU2013drop,Qualcomm2013drop,lin2013uplinkD2D,BacLi11}). In this paper, we propose to use the following random dropping model as a first-cut study, and leave other models (e.g.,  clustered UEs in hotspot) to future work. We assume that the D2D transmitters of the $i$th type are randomly distributed according to a homogeneous PPP $\Phi_{D_i}$ with density $\lambda_{D_i}$. The $M$ PPPs are assumed to be  independent of each other. Note that the performance of the PPP model can serve as a benchmark for more general settings. Each receiver is assumed to be randomly located around its transmitter according to a two-dimensional Gaussian distribution $N(0,\delta^2)$, with the phase uniformly distributed in~$[0, 2\pi]$, so $\delta$ parameterizes the distance between the receiver and its transmitter which is Rayleigh distributed with mean $\delta\sqrt{\frac{\pi}{2}}$~\cite{BacLi11}.   Other distance distributions can be easily incorporated into the considered  framework.

We model the BSs and cellular users in the cellular network as two further independent homogeneous PPPs, denoted by $\Phi_B$ and $\Phi_U$ with densities $\lambda_B$ and $\lambda_U$, respectively. The model can be easily extended to the case where cellular users have heterogeneous service demands.  By tuning the BS, D2D and cellular user densities, along with $\delta$, a very large class of plausible network topologies can be considered with this~framework.

\begin{figure}
\centering
\includegraphics[width=6.4cm,height=6.1cm]{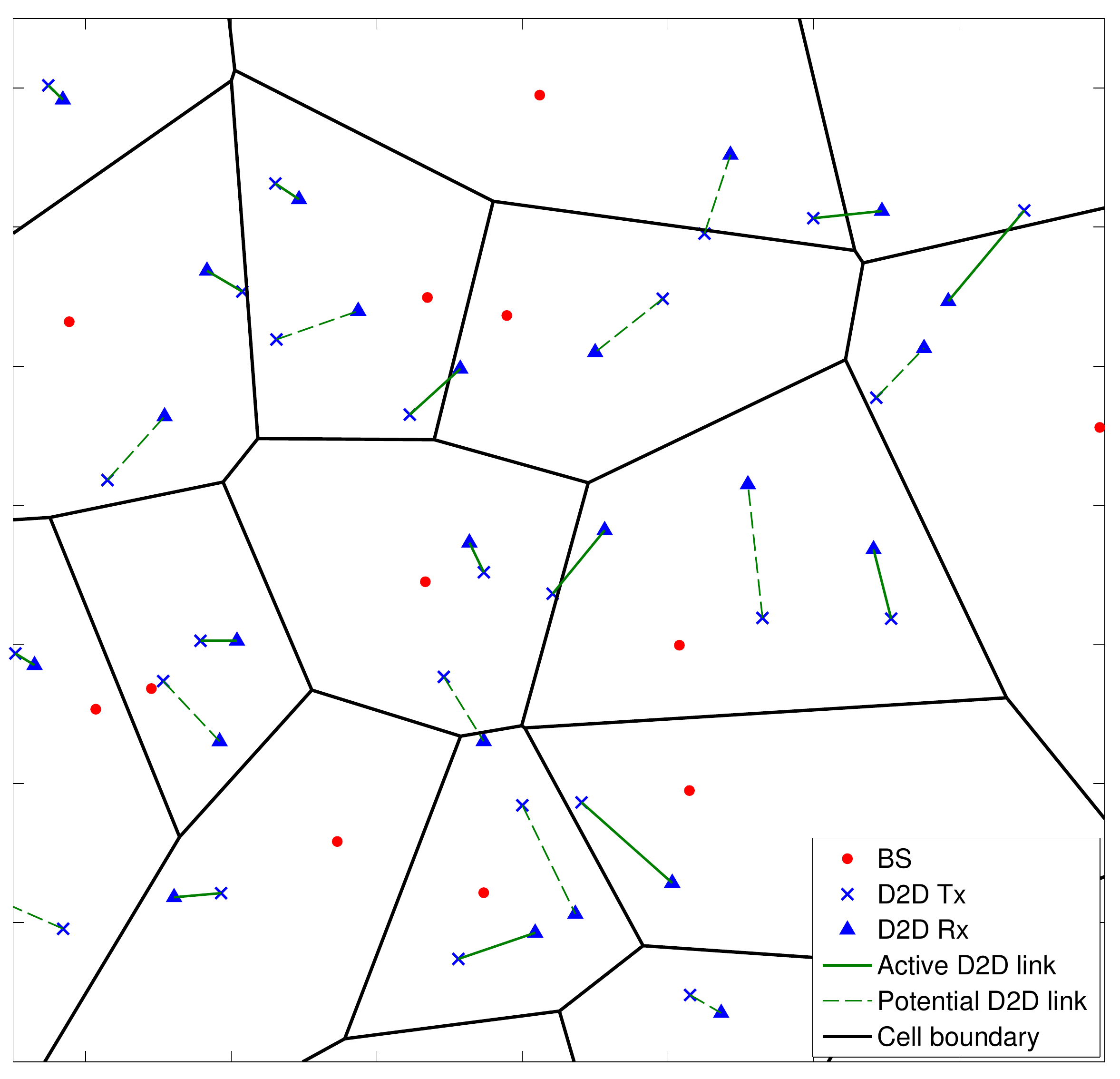}
\caption{Illustration of the network model. The red points are BSs which are deployed according to a PPP. The D2D links include both silent potential D2D links (with dashed lines) and active D2D links (with solid lines).}
\label{fig:network}
\end{figure}

\vspace{-0.4cm}
\subsection{Scheduling scheme}\label{sec:model-scheduling}
We propose to  adopt a time-frequency hopping scheme for scheduling D2D links, to randomize the occurrence of access collisions with nearby interfering UEs, and thus randomize the  strong interference~\cite{chen2010time}. 
As illustrated in Fig. \ref{fig:hopping}, the time axis is divided into consecutive operation slots. At each slot, the potential D2D links can either be active (i.e., in D2D mode, where traffic is transmitted directly between UEs) or silent (i.e., in cellular mode, where traffic is relayed via the BS), and each potential D2D link selects its operation mode  independently. For example, a potential D2D link of type $i$ would either be active   with probability $p_{t_i}\in[0, 1]$ or operate in  cellular mode with probability $1-p_{t_i}$.  As $p_{t_i}$ increases, more potential D2D links would be in D2D mode. Thus the time hopping is a tool for implementing mode selection, where the mode selection parameter~$p_{t_i}$ results in a  tradeoff between spatial reuse and additional interference. 
In the frequency domain, the $i$th type D2D links would access each subband independently with probability $p_{f_i}\in[0, 1]$. As $p_{f_i}$ increases, more frequency resources are utilized by D2D links, at the cost of increasing interference since more D2D links access the same subbands. Therefore, the frequency hopping probability $p_{f_i}$ results in a tradeoff between frequency efficiency and additional interference. 

Using the time-frequency hopping scheme, D2D links are  scheduled independently of one another in an Aloha-type fashion in both time and frequency ~\cite{Abr70}.  
The outage probability, which is defined as the probability that the SINR is less than or equal to a given threshold (i.e., $\mathbb{P}\left(\sinr \leq \beta\right)$, where $\beta$ is a predefined threshold),  increases almost linearly with the  hopping probabilities in the low outage regime, while the spatial reuse increases linearly with the number of time or frequency slots~\cite{WebYan05}. So we can adjust the outage probability by changing the time and frequency hopping probabilities, so as to meet a target outage constraint (i.e., $\mathbb{P}\left(\sinr \leq \beta\right)\leq \epsilon$, where $\epsilon$ is a predefined parameter).
Other scheduling approaches such as centralized approaches or CSMA can be adopted, but with drawbacks in both practice (e.g. high overhead) and in terms of tractability (the resulting transmitters are correlated and thus no longer a PPP).  
 
Further, we introduce a penalty assessed to potential D2D links operating in cellular mode, denoted by~$w$, to account for using both uplink and downlink time-frequency resources. A nominal value for~$w$ might be 2, because the local traffic transmitted via a BS requires to establish both the uplink and downlink transmissions, while the D2D transmission only needs to establish one link. 
The parameter $w$ can be considered as the price for D2D traffic using cellular mode, which can help adjust the load between D2D and cellular~networks.

\begin{figure}
\centering
\includegraphics[width=6.9cm,height=6cm]{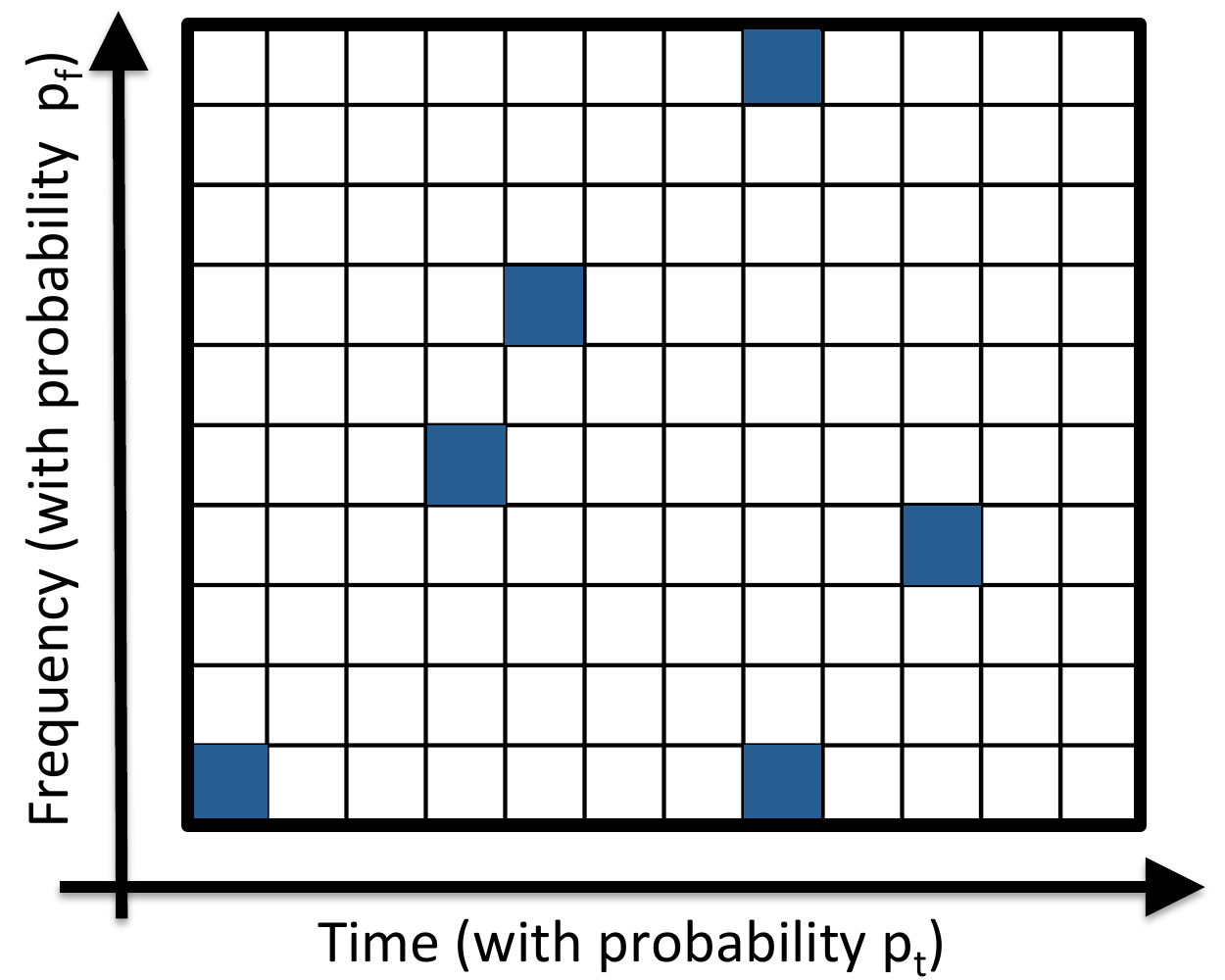}
\caption{Illustration of the time-frequency hopping scheme. The shadowed squares are the resource blocks (RBs) occupied by some active D2D links. A D2D link accesses each time slot uniformly with probability $p_t$ and accesses each subband uniformly with probability $p_f$.}
\label{fig:hopping}
\end{figure}

\vspace{-0.4cm}
\subsection{Load modeling}\label{sec:model-load}
\vspace{-0.1cm}
We assume that there are $B$ frequency slots (subbands) in the network. All potential D2D links would either operate in the dedicated network or the shared network. 
In a dedicated network,  a fraction of resources is allocated to D2D links, denoted by $\theta$, while the rest is allocated to the cellular network. In contrast, in the shared network, the active D2D links share the resources with the cellular network.  We assume that the UEs are associated with the nearest BSs, and each BS randomly allocates resource blocks (RBs) to its cellular UEs according to their service demands. The performance using random allocation is a lower bound on more sophisticated scheduling schemes (e.g., ones which are channel or SINR-dependent), and the consideration of such schemes is left to future work.

The cellular users are assumed to have the same resource requirement, denoted by $b_{C}$. The resource requirement for the $i$th type D2D links is denoted by $b_{D_i}$. In this paper, we consider the resource requirement in terms of the number of subbands for tractability \cite{LinAnd12}. The ratio of total subbands $B$ to the resource requirement of UEs represents different service demand scenarios (e.g., $b_{C}=b_{D_i}=B$ for a heavily loaded network). Let $N_{C}$ and $N_{D_i}$ be the number of original cellular users and of type $i$ potential D2D links in cellular mode, respectively. Let $B_C$ denote the number of available subbands for the cellular network, where $B_C=(1-\theta)B$ in the dedicated network and~$B_C=B$ in the shared network. The cell is lightly loaded if $B_C/(b_{C}N_{C}+\sum_{i=1}^{M}(1-p_{t_i})b_{D_i}N_{D_i})>1$, and fully loaded otherwise. In the former case, the BS only transmits on a subset of the subbands, which are called normal RBs, while the other RBs are left blank (i.e., the BS is not transmitting in the corresponding RBs). 
In the latter case, some users have to be blocked (or all UEs are admitted by a cell but can only obtain a fraction of time slots). The admission probability (or fraction of time slots) can be obtained by~$B_C/(b_{C}N_{C}+\sum_{i=1}^{M}(1-p_{t_i})b_{D_i}N_{D_i})$~\cite{LinAnd12},  which is essentially the ratio of the number of available subbands in the cell to the number of subbands needed by cellular users.


\vspace{-0.4cm}
\subsection{Channel model}
In this paper, we assume that the transmission powers are fixed at $P_D$ for D2D transmitters and $P_B$ for cellular BSs. General attenuation functions can be adopted, but we focus on the standard power law attenuation function~$l(d)=d^{-\alpha}$, where $d$ is the distance from the transmitter to a receiver, and $\alpha$ is the path loss exponent. 
 We assume all links experience independent Rayleigh fading. Shadowing is not explicitly modeled, but is already captured by the randomness of PPP in some sense, e.g., \cite{BlaBar12} showed that a grid BS model with fairly strong (standard deviation greater than 10dB) log-normal shadowing is nearly equivalent to a PPP model without shadowing.

The notations are summarized in Table \ref{tb:notation}.

\begin{table}
\small
\caption{Notation Summary}\label{tb:notation}
\begin{center}
\begin{tabular}{|c||c|}
\hline
Notation & Description\\
\hline
$\Phi_{D_i}$ & PPP of type $i$ D2D links\\
\hline
$\Phi_B$,$\Phi_U$ & PPP of BSs, cellular UEs\\
\hline
$\lambda_{D_i}$ & Density of type $i$ D2D links\\
\hline
$\lambda_B$, $\lambda_U$ & Density of BSs, cellular UEs\\
\hline
 $M$ & Number of D2D types\\
 \hline
 $\delta$ & \tabincell{c}{Parameter of distance between\\ D2D transmitter and its receiver}\\
 \hline
$p_{f_i}$ & Frequency hopping probability\\
\hline
 $p_{t_i}$ & Time hopping probability\\
\hline
$w$ & Penalty for potential D2D links in cellular mode\\
\hline
$\theta$ & Fraction of resource allocated to D2D\\
\hline
$B$ & Total frequency subbands\\
\hline
$b_{D_i}$ & Service demand of type $i$ D2D links\\
 \hline
$b_C$ & Service demand of cellular users\\
\hline
$P_B$& Transmit power of BSs\\
\hline
$P_D$& Transmit power of D2D transmitters\\
\hline
$\sigma^2$ & Noise  power \\
\hline
$\alpha$ & Path loss exponent\\
\hline
$\rho^{(O)}$, $\rho^{(S)}$ & \tabincell{c}{Fraction of normal RBs in the dedicated\\ and shared network, respectively}\\
\hline
$p_a^{(O)}$, $p_a^{(S)}$ & \tabincell{c}{Admission probability in the dedicated\\ and  shared network, respectively}\\
\hline
$P_D^{(O)}$, $P_D^{(S)}$ & \tabincell{c}{Coverage probability of D2D links in the\\ dedicated and shared network, respectively}\\   
\hline
$P_C^{(O)}$, $P_C^{(S)}$ & \tabincell{c}{Coverage probability of cellular UEs in the\\ dedicated and shared network, respectively}\\
\hline
$R_D^{(O)}$, $R_D^{(S)}$ & \tabincell{c}{Rate of D2D links in the dedicated\\ and shared network, respectively}\\
\hline
$R_C^{(O)}$, $R_C^{(S)}$ & \tabincell{c}{Rate of cellular UEs in the dedicated \\and shared network, respectively}\\
\hline
\end{tabular}
\end{center}
\end{table}
\section{Analysis of the Dedicated  Network}\label{sec:overlay}
In this section, we investigate the key performance metrics in the dedicated network.

\vspace{-0.4cm}
\subsection{SINR Distribution}
Without loss of generality, due to the stationarity of $\Phi_D$, we conduct analysis on a typical D2D receiver located at the origin, whose transmitter is active. The location of the typical transmitter  is denoted by~$X_0$. For simplicity, we denote D2D links by the location of their transmitters (e.g., the typical link is called link~$X_0$). According to the load model described in Section \ref{sec:model-load}, the network has blank RBs when it is under-loaded. To get the average fraction of blank RBs, we first find the average load per cell. The average coverage area of a BS  is $\frac{1}{\lambda_B}$ \cite{Bac09}. Therefore, the average numbers of cellular users and of $i$th-type D2D links in cellular mode in a cell are~$\lambda_U/\lambda_B$ and~$(1-p_{t_i})\lambda_{D_i}/\lambda_B$, respectively. Recalling the available fraction of subbands for the cellular network is~$(1-\theta)B$, the average fraction of normal RBs sent by a BS is approximated~by
\begin{equation}\label{eq:rho-overlay}
\rho^{(O)}  \approx \min\{\frac{b_C\lambda_U+\sum_{i=1}^{M}(1-p_{t_i})b_{D_i}\lambda_{D_i}}{\lambda_B(1-\theta)B}, 1\}.
\end{equation}
Assuming that BSs randomly allocate RBs to cellular UEs, the set of active interfering BSs at a typical RB in the dedicated network, denoted by~$\tphi_B$, can be considered as a thinning process from the baseline BS process $\Phi_B$, which is approximated by a PPP with density $\rho^{(O)} \lambda_B$. This approximation is validated in Section \ref{sec:simulation}, where we observe that analysis and simulation results are in good agreement.

Adopting time-frequency hopping, the interfering D2D transmitters are those  which access the same time slots and subbands. We denote the set of interfering transmitters of the $i$th type by $\tphi_{D_i}$, which is a thinning point process from the PPP $\Phi_{D_i}$. Based on the Thinning Theorem of PPP \cite{Bac09}, the thinning process $\tphi_{D_i}$ is a PPP with density $p_{t_i} p_{f_i} \lambda_{D_i}$. Applying the superposition of PPPs \cite{Bac09}, the set of interfering transmitters can be considered as a single PPP $\tphi_D$ with density $\tilde{\lambda}_{D}=\sum_{i=1}^{M}p_{t_i} p_{f_i} \lambda_{D_i}$. 

The SINR of the typical D2D link is
\vspace{-0.2cm}
\begin{equation*}
\sinr = \frac{P_Dh_0 |X_0|^{-\alpha}}{I_{\tphi_D}+\sigma^2},
\end{equation*}
where $I_{\tphi_D} = \sum_{X_i\in\tphi_D \setminus X_0} P_D h_i |X_i|^{-\alpha}$ is the interference from other D2D users,  and $\sigma^2$ is the noise power. The  SINR complementary cumulative distribution function (CCDF) of the D2D links, also known as the coverage probability, is given by Proposition \ref{theo:cdf-d2d}. 

\begin{prop}\label{theo:cdf-d2d}
The SINR distribution of D2D links in the dedicated network is 
\begin{equation}\label{eq:cdf-d2donly}
\begin{aligned}
&\mathbb{P}_D^{(O)}(\beta)\overset{\triangle}{=}\mathbb{P}\left(\text{SINR}>\beta\right)\\
=&\int_0^\infty e^{-\beta P_D^{-1}\sigma^2 v^{\alpha}}\mathcal{L}_{I_{\tphi_D}}(\beta P_D^{-1} v^{\alpha}) \frac{v e^{-\frac{v^2}{2\delta^2}}}{\delta^2}dv.
\end{aligned}
\end{equation}
where the Laplace transform of the interference from D2D transmitters is
\begin{equation}\label{eq:lap-d2donly}
\mathcal{L}_{I_{\tphi_D}}(s)=\exp\left( - \tilde{\lambda}_{D} \frac{2\pi^2 /\alpha}{\sin\left(2\pi/\alpha\right)}(sP_{D})^{\frac{2}{\alpha}}  \right),
\end{equation}
\end{prop} 
\begin{proof}
See Appendix \ref{pf:cdf-d2d}.
\end{proof}

The derived SINR distribution indicates that the time and frequency hopping impact the SINR of D2D links in a product term $p_{t_i}p_{f_i}$ in $\tilde{\lambda}_D$. That is, as long as the product $p_{t_i}p_{f_i}$ is a constant, no matter how the time hopping or the frequency hopping is changed, the network performance will be the same. There is a tradeoff between the density and the performance of the links in D2D mode: as more and more potential D2D transmitters attempt to transmit, though the density of active links increases, the interference increases and thus the SINR of active D2D links decreases.

With the assumption of nearest-BS association, the distribution of the distance between a user and its associated BS, denoted by $r$, is $f_r(r)=e^{-\lambda_B\pi r^2}2\lambda_B\pi r$ \cite{StoKen87}. In the dedicated network, there is no D2D-cellular interference. Therefore, we can leverage the analytical results of the cellular network in~\cite{AndBac10} to evaluate the cellular network performance, where the SINR CCDF of cellular users is given by Proposition~\ref{prop:sinr-cellonly}.
\begin{prop}\label{prop:sinr-cellonly}
The SINR distribution of a typical cellular user in the dedicated network is
\begin{equation}\label{eq:cdf-cellonly}
\begin{aligned}
&\mathbb{P}_C^{(O)}(\beta)\overset{\bigtriangleup}{=}\mathbb{P}(\sinr>\beta)\\
=&\int_0^\infty e^{-\beta P_{B}^{-1} \sigma^2r^{\alpha}}  \mathcal{L}_{I_{\tphi_B}}(\beta P_{B}^{-1} r^{\alpha}) e^{-\lambda_B\pi r^2}2\lambda_B\pi rdr.
\end{aligned}
\end{equation}
where the Laplace transform of interference from the  cellular UEs is 
\begin{equation}\label{eq:lap-cd-bb}
\mathcal{L}_{I_{\tphi_B}}(\beta P_{B}^{-1} r^{\alpha})=\exp\left( -2\pi\rho^{(O)}\lambda_B r^2 H_1(\beta, \alpha)\right),
\end{equation}
and $H_1(\beta, \alpha)=\int_1^\infty \frac{x}{1+\beta^{-1}x^{\alpha}}dx$.
\end{prop}
\begin{proof}
The proposition is shown in \cite{AndBac10}. For completeness, we provide the details as follows.
Denoting~$s=\beta P_{B}^{-1} r^{\alpha}$,  given the distance between user and its closest BS $r$, the conditional coverage probability  is
$\mathbb{P}(\sinr (r)>\beta \mid r)= \exp(-s\sigma^2)\mathcal{L}_{I_{\tphi_B}}(s)$. Denoting the BS serving the typical user by $B_0$, the Laplace transform of interference from other BSs $I_{\Phi_B}$ is 
\begin{equation*}
\begin{aligned}
\mathcal{L}_{I_{\tphi_B}}(s)=&\mathbb{E} \left[ \exp\left(-s \sum_{B_i\in \tphi_{B}\setminus B_0} P_{B} h_i r^{-\alpha} \right)\right]\\
\overset{(a)}{=}&\exp\left( -2\pi\rho^{(O)}\lambda_B \int_r^\infty \frac{u}{1+s^{-1}P_{B}^{-1}u^{\alpha}}du\right)\\
=&\exp\left( -2\pi\rho^{(O)}\lambda_B r^2 H_1(\beta, \alpha)\right),
\end{aligned}
\end{equation*}
where the lower limit of integral $r$ in $(a)$ follows from the assumption that the user is associated with its closest BS, and the last equality is obtained by letting $x=u/r$.
\end{proof}

Considering the special case where the network is interference-limited (i.e., the thermal noise is ignored), the above results can be further simplified:
\begin{cor}\label{cor:sinr-d2donly-nonoise}
When $\sigma^2\rightarrow 0$, the SINR distributions of D2D links and cellular users, respectively, are
\begin{equation}\label{eq:cdf-d2donly-nonoise}
\mathbb{P}_D^{(O)}\left(\beta\right)=\frac{1}{1+2\delta^2 \tilde{\lambda}_{D} \frac{2\pi^2 /\alpha}{\sin\left(2\pi/\alpha\right)}\beta^{\frac{2}{\alpha}}},
\end{equation}
\begin{equation}\label{eq:cdf-cellonly-nonoise}
\mathbb{P}_C^{(O)}(\beta) = \frac{1}{2 \rho^{(O)} H_1(\beta, \alpha) + 1}.
\end{equation}
\end{cor}

\vspace{-0.5cm}
\subsection{Rate Analysis in the Dedicated Network}

In this section, we analyze the average achievable rates of cellular users and D2D links in the dedicated network. By treating the interference as noise, we use Shannon's capacity formula to approximate the rate, i.e., $W\log_2(1+\sinr)$, where $W$ is the available bandwidth. Assuming the fraction of time slots is $T$, the long-term rate becomes $R=TW\log_2(1+\sinr)$, where $TW$ can be considered as the total available fraction of RBs. 

Recall that the admission probability (i.e. available fraction of time slots) of cellular UEs is 
\begin{equation}
p_a^{(O)}=\min\left\{\frac{(1-\theta)B}{b_C\bar{N}_C+\sum_{i=1}^{M} b_{D_i}\bar{N}_{D_i}},\  1\right\},
\end{equation}
where $\bar{N}_C$ and $\bar{N}_{D_i}$ are the expected number of cellular users and of type $i$ D2D links in cellular mode in the typical user associated cell, respectively. Note that a random UE is more likely to be associated with a cell which has a larger coverage area. Denoting the BS serving the typical UE by $B_0$,  the expected coverage area of BS $B_0$ is larger than $1/\lambda_B$, known as Feller's paradox \cite{BacGlo00}. The average coverage area of BS $B_0$ is instead given by $9/(7\lambda_B)$ \cite{FerNed07,SinAnd12}. Therefore, similar to~\cite{LinAnd12}, the admission probability can be approximated to
\begin{equation}\label{eq:pa-overlay}
p_a^{(O)}\approx \min\left\{\frac{7(1-\theta)B\lambda_B}{9\left(b_C\lambda_U+\sum_{i=1}^{M}b_{D_i}(1-p_{t_i})\lambda_{D_i}\right)},\  1\right\}.
\end{equation}

Recalling that $w$ is the price for a D2D link operating in cellular mode, the average rates of cellular users and D2D links are given in Theorem \ref{theo:rate-overlay}.

\begin{theo}\label{theo:rate-overlay}
The average achievable rates of a typical cellular user and a D2D link of the $i$th type, respectively, are
\begin{equation}\label{eq:rate-overlay-celli}
R_{C_i}^{(O)}=b_{C_i}p_a^{(O)}\int_0^\infty\frac{\log_2(e)}{(\beta+1)}\mathbb{P}_C^{(O)}\left(\beta\right)d\beta,
\end{equation}
\begin{equation}\label{eq:rate-overlay-d2di}
\begin{aligned}
R_{D_j}^{(O)}=&\min\{p_{f_j}\theta B, b_{D_j} \} p_{t_j} \int_0^\infty\frac{\log_2(e)}{(\beta+1)}\mathbb{P}_D^{(O)}\left(\beta\right)d\beta \\
&+ \frac{b_{D_j}}{w}(1-p_{t_j})p_a^{(O)}\int_0^\infty\frac{\log_2(e)}{(\beta+1)}\mathbb{P}_C^{(O)}\left(\beta\right)d\beta,
\end{aligned}
\end{equation}
where $\mathbb{P}_C^{(O)}\left(\beta\right) $ and  $\mathbb{P}_D^{(O)}\left(\beta\right) $ are given in (\ref{eq:cdf-cellonly}) and (\ref{eq:cdf-d2donly}), respectively. Further, we can get the average rate of a typical D2D link by
\begin{equation}\label{eq:rate-d2donly}
\begin{aligned}
R_{D}^{(O)}=&\sum_{j=1}^{M} \frac{\lambda_{D_i}}{\lambda_D} R_{D_j}^{(O)}.
\end{aligned}
\end{equation}
\end{theo}
\begin{proof}
According to Shannon's capacity formula, the long term rate of a typical cellular user is
\begin{equation*}
\begin{aligned}
R_{C_i}^{(O)}=&b_{C_i}p_a^{(O)}\mathbb{E}\left[\log_2\left(1+\sinr \right)\right]\\
=&b_{C_i}p_a^{(O)}\int_0^\infty\mathbb{P}\left(\sinr>2^t-1\right)dt\\
=&b_{C_i}p_a^{(O)}\int_0^\infty\frac{\log_2(e)}{(\beta+1)}\mathbb{P}\left(\sinr>\beta\right)d\beta,
\end{aligned}
\end{equation*}
where  we let $2^t-1=\beta$ in the last equality. 

A typical D2D link can be either in D2D mode or cellular mode, and thus the average rate of a typical D2D link can be calculated according to 
\begin{equation*}
\begin{aligned}
R_{D_j}^{(O)}&=\mathbb{P}(\text{D2D mode})\mathbb{E}\left[R^{(O)}_{\text{D2D mode}}\right]\\
&+\mathbb{P}(\text{cellular mode})\frac{1}{w}\mathbb{E}\left[R^{(O)}_{\text{cellular mode}}\right],
\end{aligned}
\end{equation*}
where we approximate the rate obtained in  cellular mode  by the rate  obtained in the downlink system  for tractability: 
\begin{equation*}
\begin{aligned}
&\mathbb{E}\left[R^{(O)}_{\text{cellular mode}}\right] \\
 =&\min\left\{\mathbb{E}\left[R^{(O)}_{\text{cellular mode in DL}}\right], \mathbb{E}\left[R^{(O)}_{\text{cellular mode in UL}}\right]\right\}\\
\approx &\mathbb{E}\left[R^{(O)}_{\text{cellular mode in DL}}\right].
\end{aligned}
\end{equation*}
Then the D2D rates in D2D mode and in cellular mode are, respectively,
\begin{equation}\label{eq:d2dmoderate}
\begin{aligned}
\mathbb{E}\left[R^{(O)}_{\text{D2D mode}}\right] &=\min\{p_{f_j} \theta B,b_{D_j}\}\int_0^\infty\frac{\log_2(e)}{(\beta+1)}\mathbb{P}_D^{(O)}\left(\beta\right)d\beta,
\end{aligned}
\end{equation}
and 
\begin{equation}\label{eq:cellmoderate}
\begin{aligned}
\mathbb{E}\left[R^{(O)}_{\text{cellular mode}}\right]&=b_{D_j} p_a^{(O)}\int_0^\infty\frac{\log_2(e)}{(\beta+1)}\mathbb{P}_C^{(O)}\left(\beta\right)d\beta.
\end{aligned}
\end{equation}
%
Note that in rate derivation in this paper, we assume that the number of users associated with the BS serving the typical link and the SINR distribution of the typical link are independent, and thus plugging~(\ref{eq:cdf-d2donly}) and (\ref{eq:cdf-cellonly}) into (\ref{eq:d2dmoderate}) and (\ref{eq:cellmoderate}), respectively, the proof is complete.
\end{proof}


\section{Analysis of the Shared Network}\label{sec:underlay}
In this section, we turn our attention to  the shared network, where the time and frequency slots are reused between D2D and cellular networks, and thus there is D2D-cellular interference.

\vspace{-0.5cm}
\subsection{SINR Distribution of D2D links} 
As in Section \ref{sec:overlay}, let $\tphi_D$ be the set of interfering D2D links, which is a PPP with density $\tilde{\lambda}_D$. The average fraction of normal RBs in the shared network is approximated by
\begin{equation}
\rho^{(S)} \approx \min\{\frac{b_C\lambda_U+\sum_{i=1}^{M}b_{D_i}(1-p_{t_i})\lambda_{D_i}}{\lambda_BB}, 1\}.
\end{equation}

We again consider a typical active D2D receiver located at the origin. Taking into account now the interference from both cellular and D2D networks, the SINR of a typical active D2D link is
\begin{equation*}
\textrm{SINR}_D=\frac{P_{D} h_0 |X_0|^{-\alpha}}{I_{\tphi_D}+I_{\tphi_B}+\sigma^2},
\end{equation*}
where $X_0$ is the location of the typical transmitter, the interference from D2D transmitters is $I_{\tphi_D}=\sum_{X_i\in\tphi_D\setminus X_0} P_{D} h_i |X_i|^{-\alpha}$, and the interference from BSs is $I_{\tphi_B}=\sum_{B_i\in \tphi_{B}} P_{B} h_i |B_i|^{-\alpha}$.

\begin{theo}\label{theo:cdf-cd-d2d}
The SINR distribution of an active D2D link in the shared network is 
\begin{equation}\label{eq:cdf-cd-d2d}
\begin{aligned}
&\mathbb{P}_D^{(S)} (\beta)\overset{\triangle}{=}\mathbb{P}\left(\text{SINR}_D>\beta\right)\\
=&\int_0^\infty e^{-s\sigma^2}  \mathcal{L}_{I_{\tphi_D}}(s)\mathcal{L}_{I_{\tphi_B}}(s)\frac{v}{\delta^2}\exp{(-\frac{v^2}{2\delta^2})}dv,
\end{aligned}
\end{equation}
where $s=\beta P_{D}^{-1} v^{\alpha}$, $\mathcal{L}_{I_{\tphi_D}}(s)$ can be calculated according to~(\ref{eq:lap-d2donly}), and
\begin{equation}\label{eq:lap-cd-bd}
\mathcal{L}_{I_{\tphi_B}}(s)=\exp\left( -2\pi\rho^{(S)}\lambda_B v^{2} H_0(\beta, \alpha)\right),
\end{equation}
where $H_0(\beta, \alpha)=\int_0^\infty \frac{x}{1+\beta^{-1}P_{D}/P_{B}x^{\alpha}}dx$.
\end{theo}
\begin{proof}
Given the distance between D2D transmitter and its receiver, denoted by $v$, the conditional coverage probability  is $\mathbb{P}(\sinr (v)>\beta \mid v)\overset{(a)}{=} \exp(-s\sigma^2)\mathcal{L}_{I_{\tphi_D}}(s)\mathcal{L}_{I_{\tphi_B}}(s)$,
where $(a)$ follows from the fact that $h_0$ is Rayleigh fading and  $I_{\tphi_D}$ is independent of $I_{\tphi_B}$. The Laplace transform of $I_{\tphi_D}$ can be calculated according to (\ref{eq:lap-d2donly}). 
Similarly, we have
\begin{equation*}
\begin{aligned}
\mathcal{L}_{I_{\tphi_B}}(s)=&\mathbb{E}\left[ \prod _{Z_i\in{\tphi_B}} \frac{1}{1+sP_{B} |B_i|^{-\alpha}}\right]\\
=&\exp\left(  \int_0^\infty \frac{-2\pi\rho^{(S)}\lambda_B r}{1+\beta^{-1}P_{D}/P_{B}(r/v)^{\alpha}}dr\right)\\
=&\exp\left( \int_0^\infty \frac{-2\pi\rho^{(S)}\lambda_B v^{2}  x}{1+\beta^{-1}P_{D}/P_{B}x^{\alpha}}dx\right),
\end{aligned}
\end{equation*}
where the last equality is obtained by letting $x=r/v$. Letting $H_0(\beta, \alpha)=\int_0^\infty \frac{x}{1+\beta^{-1}P_{D}/P_{B}x^{\alpha}}dx$, the proof is complete.
\end{proof}

Theorem \ref{theo:cdf-cd-d2d} shows that for any given SINR threshold, the coverage probability of D2D links is monotonically decreasing as the access probabilities $p_{t_i}$ and/or $p_{f_i}$ increase, due to the increasing interference from the D2D network. On the other hand, the relationship between the MAC protocol and average rate is more subtle, and is discussed in Section~\ref{sec:rate}.
\vspace{-0.5cm}
\subsection{SINR Distribution of Cellular Users}
\vspace{-0.1cm}
In this section, we conduct analysis on a typical cellular  UE  in the shared network,  which is assumed to be located at the origin. 

\begin{theo}\label{theo:cdf-cd-cell}
The SINR distribution of a typical cellular user in the shared network is 
\begin{equation}\label{eq:cdf-cd-cell}
\begin{aligned}
&\mathbb{P}_C^{(S)}(\beta)\overset{\triangle}{=}\mathbb{P}\left(\text{SINR}>\beta\right)\\
=&\int_0^\infty e^{-s\sigma^2}  \mathcal{L}_{I_{\tphi_D}}(s)\mathcal{L}_{I_{\tphi_B}}(s) e^{-\lambda_B\pi r^2}2\lambda_B\pi rdr,
\end{aligned}
\end{equation}
where $s=\beta P_{B}^{-1} r^{\alpha}$, and the Laplace transform of interference from D2D links $\mathcal{L}_{I_{\tphi_D}}(s)$ can be calculated according to~(\ref{eq:lap-d2donly}). The Laplace transform of interference from the cellular network $\mathcal{L}_{I_{\tphi_B}}(s) $ can be obtained by  (\ref{eq:lap-cd-bb}), where $\rho^{(O)}$ should be replaced by $\rho^{(S)}$.
\end{theo}
\vspace{-0.4cm}
\begin{proof}
We omit the proof as it is similar to the proof of Theorem \ref{theo:cdf-cd-d2d}. 
\end{proof}

Ignoring thermal noise, the results can be again significantly simplified.
\begin{cor}\label{cor:cdf-cd-spec}
When $\sigma^2\rightarrow 0$, the SINR distributions of D2D links and of cellular users are, respectively,
\begin{equation}\label{eq:cdf-cd-dspec}
\begin{aligned}
\mathbb{P}_{D}^{(S)}(\beta)=\frac{1}{ 2\delta^2 \tilde{\lambda}_D\kappa \pi \beta^{\frac{2}{\alpha}}  +4\delta^2\pi\rho^{(S)}\lambda_B  H_0(\beta, \alpha)+1},
\end{aligned}
\end{equation}
and
\begin{equation}\label{eq:cdf-cd-cspec}
\begin{aligned}
\mathbb{P}_{C}^{(S)}(\beta)=\frac{1 }{\frac{\tilde{\lambda}_D}{\lambda_B} \kappa (\beta \frac{P_D}{P_B})^{\frac{2}{\alpha}} + 2 \rho^{(S)} H_1(\beta,\alpha) +1},
\end{aligned}
\end{equation}
where $\kappa=\frac{2\pi /\alpha}{\sin\left(2\pi/\alpha\right)}$.
\end{cor}

\begin{proof}
See Appendix \ref{pf-cor:cdf-cd-spec}.
\end{proof}

According to the above analysis, we can see that the coverage probabilities of  D2D links and of cellular UEs are both monotonically decreasing functions of $p_{t_i}$ and $p_{f_i}$, due to the increasing interference as more D2D links access to the same resource block. 

\vspace{-0.5cm}
\subsection{Rate Analysis in the Shared Network }\label{sec:rate}
Similar to the dedicated system, the admission probability of cellular users  is
\begin{equation}\label{eq:pa-underlay}
p_a^{(S)}\approx\min\left\{\frac{7B\lambda_B}{9\left(b_C\lambda_U+\sum_{i=1}^{M}b_{D_i}(1-p_{t_i})\lambda_{D_i}\right)},\  1\right\}.
\end{equation}

The average rates of cellular users and D2D links in the shared network are given in Proposition \ref{theo:rate-underlay}.
\begin{prop}\label{theo:rate-underlay}
The average achievable rate of a cellular user is
\begin{equation}\label{eq:rate-underlay-celli}
R_{C}^{(S)}= b_{C}p_a^{(S)}\int_0^\infty\frac{\log_2(e)}{(\beta+1)}\mathbb{P}_C^{(S)}\left(\beta\right)d\beta,
\end{equation}
where  $\mathbb{P}_C^{(S)}\left(\beta\right) $ is given by (\ref{eq:cdf-cd-cell}).
The average achievable rate of a type $i$ D2D link  is
\begin{equation}\label{eq:rate-underlay-d2di}
\begin{aligned}
R_{D_i}^{(S)}&=p_{t_i}\min\{p_{f_i}B,b_{D_i}\}\int_0^\infty\frac{\log_2(e)}{(\beta+1)}\mathbb{P}_D^{(S)}\left(\beta\right)  d\beta \\
&+\frac{b_{D_i}}{b_C w}(1-p_{t_i})R_{C}^{(S)},
\end{aligned}
\end{equation}
where $w$ is the penalty for the potential D2D links transmitting by BSs, and $\mathbb{P}_D^{(S)}\left(\beta\right) $ is given by (\ref{eq:cdf-cd-d2d}).
\end{prop}

Then we can get the average rate of  a typical D2D link as $R_D^{(S)} = \sum_{i=1}^{M} \frac{\lambda_{D_i}}{\lambda_D}R_{D_i}^{(S)}$.

\begin{cor}\label{cor:rateLB}
We further have lower bounds on the rates
\begin{equation}
R_C ^{(S)}\geq R_{Cl}^{(S)}=\sup_\beta\  b_Cp_a^{(S)} \log_2(1+\beta)\mathbb{P}_C^{(S)}\left( \beta\right),
\end{equation}
and $R_{D_i}^{(S)} \geq R_{Dl_i}^{(S)}$, where 
\begin{equation}
\begin{aligned}
 R_{Dl_i}^{(S)}&=\sup_\beta \left(p_{t_i}\min\{p_{f_i}B,b_{D_i}\}\log_2(\beta+1)\mathbb{P}_D^{(S)}\left(\beta\right)  \right)\\
&+\frac{b_{D_j}}{b_C w}(1-p_{t_i}) R_{Cl}^{(S)} ,
\end{aligned}
\end{equation}
where $\mathbb{P}_C^{(S)}(\beta)$ and $\mathbb{P}_D^{(S)}(\beta)$ can be calculated according to (\ref{eq:cdf-cd-cell}) and (\ref{eq:cdf-cd-d2d}), respectively.
\end{cor}
\begin{proof}
Denoting $\Gamma=\sinr$, we have the following inequality for any $\beta$.
\begin{equation*}
\begin{aligned}
\mathbb{E}\left[\log_2(1+\Gamma)\right]=&\mathbb{P}(\Gamma>\beta)\mathbb{E}\left[\log_2(1+\Gamma)\mid \Gamma>\beta\right]\\
&+\mathbb{P}(\Gamma\leq\beta)\mathbb{E}\left[\log_2(1+\Gamma)\mid \Gamma\leq\beta\right]\\
\geq & \mathbb{P}(\Gamma>\beta)\mathbb{E}\left[\log_2(1+\Gamma)\mid \Gamma>\beta\right]\\
\geq &\mathbb{P}(\Gamma>\beta) \log_2(1+\beta).
\end{aligned}
\end{equation*}
Therefore, we have $\mathbb{E}\left[\log_2(1+\sinr)\right]\geq \sup_\beta  \mathbb{P}(\sinr>\beta) \log_2(1+\beta)$.
\end{proof}
The above lower bounds can also be extended to the dedicated network.

According to the above analysis, when the frequency hopping probability $p_{f_i}$ increases, the rate of a typical cellular user decreases, because the interference from D2D links increases. Thus the cellular rate is a monotonic function of $p_{f_i}$. As for the time hopping, when the time hopping probability $p_{t_i}$ increases, more potential D2D links would operate in  D2D mode. On the one hand, the interference from D2D links increases as time hopping probability increases, which leads to a lower SINR for the cellular links; on the other hand, the cellular links benefit from D2D offloading, since more resources would be available for the remaining cellular links. Therefore, it is difficult to determine the impact of time hopping on the rate of cellular users.
As for the rate of a typical  D2D link, it is even more difficult to explore the impact of the MAC protocol, because both time and frequency hopping result in the tradeoff between  resource efficiency and additional interference. It is not {\it a priori} clear whether larger time and frequency hopping probabilities would be beneficial or not.
However, by changing variables, we can get the optimal solution of at least one variable and thus reduce the dimensions of the optimization problem. We explore these issues in detail in the next section.

 \section{Optimization of the D2D-enabled Cellular Network}\label{sec:opt}
Based on the derived analytical results, we now turn our attention to the optimization of network performance. As in \cite{WebYan05,Bac09,HuaLau09}, we study the utility maximization in terms of average long-term rates in  the interference limited network (i.e., $\sigma^2 \rightarrow 0$) for simplicity. 
 
\vspace{-0.4cm}
\subsection{Optimization of the Dedicated Network}
 The utility functions of a cellular user and of a type $i$ D2D link are denoted by $U_C(R_{C})$ and $U_D(R_{D_i})$, respectively, where $U_D(\cdot)$ and $U_C(\cdot)$ are continuously differentiable, non-decreasing, and concave functions~\cite{StaWic09}. The optimization problem can be formulated as
\begin{equation}\label{eq:opt-utility-overlay}
\begin{aligned}
\max\limits_{p_t, p_f} \quad & \sum_{i=1}^{M}\lambda_{D_i} U_D(R_{D_i}^{(O)})+\lambda_U U_C(R_C^{(O)})\\
\text{s.t. } \quad & 0\leq p_{t_i} \leq 1, 0\leq p_{f_i} \leq 1,\ \forall i\in \Phi_D,
\end{aligned}
\end{equation}
where $R_{D_i}^{(O)}$ and $R_{C}^{(O)}$ are obtained by Theorem \ref{theo:rate-overlay}.

The next result shows that the optimal frequency hopping probability can be obtained in closed form, and is \emph{independent of the choice of utility functions}.
\begin{prop}\label{prop:overlay-pf}
For any non-decreasing utility function, the optimal frequency hopping probability for rate is $p_{f_i}^* = \min\{1, b_{D_i}/(\theta B)\}$.
\end{prop}
\begin{proof}
The objective function is a non-decreasing function of $p_{f_i}$ when $p_{f_i}\theta B\leq b_{D_i}$, and becomes monotonically decreasing when $p_{f_i}\theta B> b_{D_i}$. Therefore, the optimal frequency hopping probability is $p_{f_i}^* = \min\{1, b_{D_i}/(\theta B)\}$.
\end{proof}
The above proposition shows that the D2D network is resource limited. The larger the service demand is, the more aggressive the D2D link should be to access the frequency~bands.

Though maximization of total rate may not be a good performance metric in the sense that it has not considered fairness among UEs, it is a reasonable objective function for a first-cut investigation of the complicated hybrid network. Therefore, in the following, we focus on the linear utility function $U(x)=x$. A single tier cellular network is heavily loaded in most cases (e.g., in a typical LTE network with $B_C=10$MHz, $b_C\approx 1$MHz, and $\frac{\lambda_U}{\lambda_B}>10$, we have $B_C\lambda_B<b_C\lambda_U$). Therefore, we consider a congested network where $7B_C\lambda_B < 9 b_C\lambda_U$, and thus $\rho^{(O)}=1$ (i.e., BSs always send normal RBs) and $p_a^{(O)}=\frac{7B_C\lambda_B}{9\left(b_C\lambda_U+\sum_{i=1}^{M}b_{D_i}(1-p_{t_i})\lambda_{D_i}\right)}$.

\begin{defi}
The rate density is defined as the expected total rate of D2D links and cellular users per surface unit.
\end{defi}
For tractability, we investigate the hopping scheme to maximize the rate lower bounds given by Corollary~\ref{cor:rateLB}.  We compare the results of exact rates and their lower bounds by simulations in Section~\ref{sec:simulation}. Note that the following results can be easily extended to the cases where the Modulation and Coding Scheme (MCS) is not adaptive by setting a fixed $\beta$. Using the rate lower bounds, the rate density of dedicated and shared networks can be calculated by

\begin{equation}\label{eq:drate-overlay}
\begin{aligned}
d_{\rt}^{(O)} \overset{\triangle}{=}&\sum_{j=1}^{M}\lambda_{D_j} R_{Dl_j}^{(O)}+\lambda_U R_{Cl}^{(O)},\\
\end{aligned}
\end{equation}
and
\begin{equation}\label{eq:drate-underlay}
\begin{aligned}
d_{\rt}^{(S)}&=\sum_{j=1}^{M}\lambda_{D_j} R_{Dl_j}^{(S)}+\lambda_U R_{Cl}^{(S)},
 \end{aligned}
\end{equation}
where $R_{Dl_j}^{(O)}$, $R_{Cl}^{(O)}$, $R_{Dl_j}^{(S)}$ and $R_{Cl}^{(S)}$ are given by Corollary \ref{cor:rateLB}.


\begin{prop}\label{prop:drate-overlay}
To maximize the rate density in the dedicated network with $w\geq 1$, we have $p_{t_i}^* = 1$. On the other hand, when $w\rightarrow 0$, we have $p_{t_i}^*\rightarrow 0$. \vspace{-0.2cm}
\end{prop}
\begin{proof}
See Appendix \ref{pf-prop:drate-overlay}.
\end{proof}
Propositions \ref{prop:overlay-pf} and \ref{prop:drate-overlay} imply that both D2D and cellular networks are resource limited when the network is fully loaded. In order to utilize resources efficiently, all potential D2D links would be in D2D mode when $w\geq 1$. On the other hand, by setting $w$ small enough, the potential D2D links can be pushed to cellular mode (i.e., $p_{t_i}^*\rightarrow 0$). 

Note that analytically it is true that all potential D2D links are in D2D mode to maximize the total average rate. However, traffic channels in real cellular systems are typically not designed to operate at very low SINR (e.g., $\sinr < -6$dB) \cite{3GPPratecap}. If the average distance is very large such that the SINRs of many D2D links are smaller than $-6$dB, the optimal mode selection would be different. 
Also, maximization of different utility functions would lead to different optimal mode selections. For example, when we consider the max-min utility, the optimal time hopping would depend on the average distance between a D2D transmitter and its receiver (which is characterized by $\delta$).  As $\delta$ increases, the rates of potential D2D links operating in D2D mode decrease, and may be smaller than the rate obtained in cellular mode. Thus with an increasing probability,  the potential D2D links in D2D mode would have the minimal rates in the system. Therefore, we would push some potential D2D links to cellular mode in this example, in order to increase their rates and maximize the minimal rate (i.e., optimal time hopping probability $p_t^*<1$).

Given the optimal time hopping and frequency hopping, we investigate the optimal resource partition between D2D and cellular networks (i.e., $\theta$). Plugging $p_{f_i}^* = \min\{1, \frac{b_{D_i}}{\theta B} \}$ and $p_{t_i}^*=1$ to (\ref{eq:drate-overlay}), the objective function is a non-differentiable function of $\theta$.
We denote $\tilde{b}_i = b_{D_i}/B$ for $i=1,\cdots,M$ and $\tilde{b}_0=0$ for simplicity. Without loss of generality, we assume the sequence $\{\tilde{b}_i\}_{i=0}^{M}$ is in ascending order (i.e., $\tilde{b}_0$ is the smallest and $\tilde{b}_{M}$ is the largest). Let $\tilde{b}_L$ be the largest $\tilde{b}_i$ that is smaller than $1$. We partition the domain of $\theta$  into $\left[\tilde{b}_i, \min\left\{\tilde{b}_{i+1},1\right\}\right], i=0, \cdots, L$. On the $i$th region $\left[\tilde{b}_i, \min\left\{\tilde{b}_{i+1},1\right\}\right]$, the types of D2D links can be separated into two sets, where $\mathcal{S}_i=\{0,\cdots, i\}$ and $\mathcal{G}_i = \{i+1,\cdots, M\}$. We have $p_{f_j}^*=\frac{\tilde{b}_j}{\theta}$ for $j\in\mathcal{S}_i$, and $p_{f_j}^*=1$ for $j\in\mathcal{G}_i$. Thus, the objective function becomes a differentiable function on each partition.   Denote $A_i=B\log_2(\beta_D+1)\sum_{j\in\mathcal{S}_i} \lambda_{D_j} \tilde{b}_j $, $C_i=2\delta^2 \frac{2\pi^2 /\alpha}{\sin\left(2\pi/\alpha\right)}\beta^{\frac{2}{\alpha}}\sum_{j\in\mathcal{S}_i} \lambda_{D_j}\tilde{b}_j $, $D= \frac{7B \lambda_B}{9}\frac{\log_2(\beta_C+1)}{2  H_1(\beta_C, \alpha) + 1}$, $E_i=B\log_2(\beta_D+1) \sum_{j\in\mathcal{G}_i} \lambda_{D_j} $, and $F_i=2\delta^2 \sum_{j\in\mathcal{G}_i}  \lambda_{D_j} \frac{2\pi^2 /\alpha}{\sin\left(2\pi/\alpha\right)}\beta^{\frac{2}{\alpha}} + 1$.
Letting
\begin{equation}\label{eq:tildeb}
\tilde{b}'_i=\left\{
\begin{aligned}
& 1, \text{ if } E_i>DF_i,\\
&\sqrt{\frac{C_i(A_iF_i-E_iC_i)}{F_i^2(D F_i-E_i})}-\frac{C_i}{F_i}, \text{ otherwise},
\end{aligned}\right.
\end{equation}
we can express the optimal solution $\theta^*$ in terms of $\tilde{b}'_i$, which is given in Proposition \ref{prop:theta}.

\begin{prop}\label{prop:theta}
The optimal $\theta$ to maximize (\ref{eq:drate-overlay}) belongs to the following set:
\begin{equation}
\mathcal{O}= \left\{  \left[\tilde{b}'_{i}\right]_{\tilde{b}_i}^{\min\{1, \tilde{b}_{i+1}\}}: i=0,\cdots,L\right\},
\end{equation}
where $\tilde{b}'_{i}$ is defined as (\ref{eq:tildeb}) and $[x]_a^b$ denotes $\min\{\max\{x,a\},b\}$.
In other words, $\theta^*=\argmax_{\theta\in\mathcal{O}} d_{\rt}^{(O)}$.
\end{prop}
\begin{proof}
See Appendix \ref{pf-prop:theta}.
\end{proof}
The parameters $A_i$, $C_i$, $E_i$ and $F_i$ can be calculated  through partial sum, which leads to a computational complexity of $O(M)$ to get the set $\mathcal{O}$. Recalling that $L$ is the number of D2D types with $\tilde{b}_i<1$, the cardinality of set $\mathcal{O}$ is at most $L+1$, where $L\leq M$. Note that $M$ is generally  a small number, which implies that~$L+1$ is small, and thus the result in Proposition~\ref{prop:theta} significantly reduces the complexity compared to the brute force search. Note that $E_i-F_iD$ decreases as $\theta$ increases. We have shown in  Appendix \ref{pf-prop:theta} that the objective function is non-decreasing when $E_i-F_iD\geq 0$. Therefore, we only need to search over the domains where $E_i<F_iD$, and thus the  cardinality of the set $\mathcal{O}$ can be further reduced.

\vspace{-0.4cm}
\subsection{Optimization of the Shared Network}

In this section, we turn our attention to the optimization of the performance in the shared network. Similarly to the dedicated network, the objective is to maximize the utility function in terms of the rate lower bounds given by Corollary~\ref{cor:rateLB}. We again consider a heavily loaded network with $\rho^{(S)}=1$ and $p_a^{(S)}=\frac{7B\lambda_B}{9\left(b_C\lambda_U+\sum_{i=1}^{M}b_{D_i}(1-p_{t_i})\lambda_{D_i}\right)}$. 
Under these assumptions, we have the following conclusion.

\begin{prop}\label{prop:drate-underlay}
Given $w\geq 1$, the optimal time hopping to maximize the rate density  (\ref{eq:drate-underlay}) is $p_{t_i}^*=1,\ \forall i\in\{1,\cdots, M\}$, i.e., all potential D2D links are in D2D mode. In contrast, when $w\rightarrow 0$, we have $p_{t_i}^*\rightarrow 0$.

\end{prop}
\begin{proof}
See Appendix \ref{pf-prop:drate-underlay}.
\end{proof}
 
Similarly to the dedicated network, $w$ can be adopted as a parameter to balance load between cellular and D2D networks, by decreasing which we can push D2D links to cellular mode. 
Though it is difficult to obtain the optimal frequency hopping in closed form in a general shared network, the maximization has been reduced to a lower-dimensional problem by finding the optimal time hopping probability, and the complexity to search the optimal scheme becomes much less. Denoting  the number of possible values of $p_{t_i}$ and $p_{f_i}$ by $|p_t|$ and $|p_f|$, respectively, the complexity of brute force can be reduced from~$\mathcal{O}((|p_t|\times|p_f|)^{M})$ to $\mathcal{O}(|p_f|^{M})$ (e.g., for the case with $|p_t|=|p_f|=100$ and $M=2$, the complexity is reduced from $\mathcal{O}(10^8)$ to $\mathcal{O}(10^4)$).

\section{Performance Evaluation}\label{sec:simulation}
In this section, we provide simulation results to validate the proposed model and analytical results.  The main simulation parameters used in this paper are summarized in Table \ref{tb:simu}, unless otherwise specified. The total bandwidth,  noise power, path loss exponent, transmit power, and density of BSs are chosen based on 3GPP documents (see, e.g., \cite{3gppmoto,3gppbackhaul}). As for the other parameters, since the D2D traffic demand and its growth is not clear at this stage, the values are chosen given the best information available to us.

\begin{table}
\small
\caption{Simulation parameters}\label{tb:simu}
\begin{center}
\begin{tabular}{|c||c|}
\hline
Total bandwidth & $10$MHz\\
\hline
Number of sub-bands $B$ & $50$\\
\hline
Number  of D2D types $M$ & 2\\
\hline
\tabincell{c}{Service demand of type $i$\\D2D links $b_{D_i}$}  & 5, 15 subbands\\
\hline
Service demand of cellular users $b_C$ & 5 subbands\\
\hline
Density of BSs $\lambda_B$ & $1/500^2$ m$^{-2}$\\
\hline
Density of cellular users $\lambda_U$ & $60/500^2$ m$^{-2}$\\
\hline
\tabincell{c}{Density of type $i$ D2D links $\lambda_{D_i}$\\(same density for different types)} & $15/500^2$ m$^{-2}$\\
\hline
\tabincell{c}{Average distance between a D2D \\ transmitter and receiver $\delta\sqrt{\frac{\pi}{2}}$ } & $50$ m \\
\hline
Transmit power of BSs $P_B$ & $46$ dBm\\
\hline
Transmit power of D2D transmitters $P_D$ & $20$ dBm\\
\hline
Noise  power $\sigma^2$ & $-104$ dBm\\
\hline
Path loss exponent $\alpha$ & $3.5$\\
\hline
\end{tabular}
\end{center}
\end{table}

\vspace{-0.5cm}
\subsection{Validation of the System Model}
We validate our analysis in Figs. \ref{fig:overlaysinr} and \ref{fig:underlaysinr}. In Fig. \ref{fig:overlaysinr}, we compare the analytical SINR cumulative distribution functions (CDFs) of D2D and cellular links in dedicated networks (given in Props. \ref{theo:cdf-d2d} and \ref{prop:sinr-cellonly}) to their corresponding simulation results.  The SINR CDFs of D2D and cellular links in shared networks (given in Theorems \ref{theo:cdf-cd-d2d} and \ref{theo:cdf-cd-cell}) are shown in Fig. \ref{fig:underlaysinr}. Recall that we approximate the set of interfering BSs by a PPP with density $\rho \lambda_B$. This approximation leads to gaps between the analysis and simulation results (e.g., the gap between analysis and simulation results of the rate of cellular links in dedicated networks). However, the gaps are very small, which implies that the approximation is reasonable. From the fact that  analytical results and their corresponding simulation results are in quite good agreement,  we conclude that stochastic geometry allows us to efficiently find the approximate coverage probabilities for the D2D-enabled cellular network.

\begin{figure}
\centering
\includegraphics[width=8cm, height=6cm]{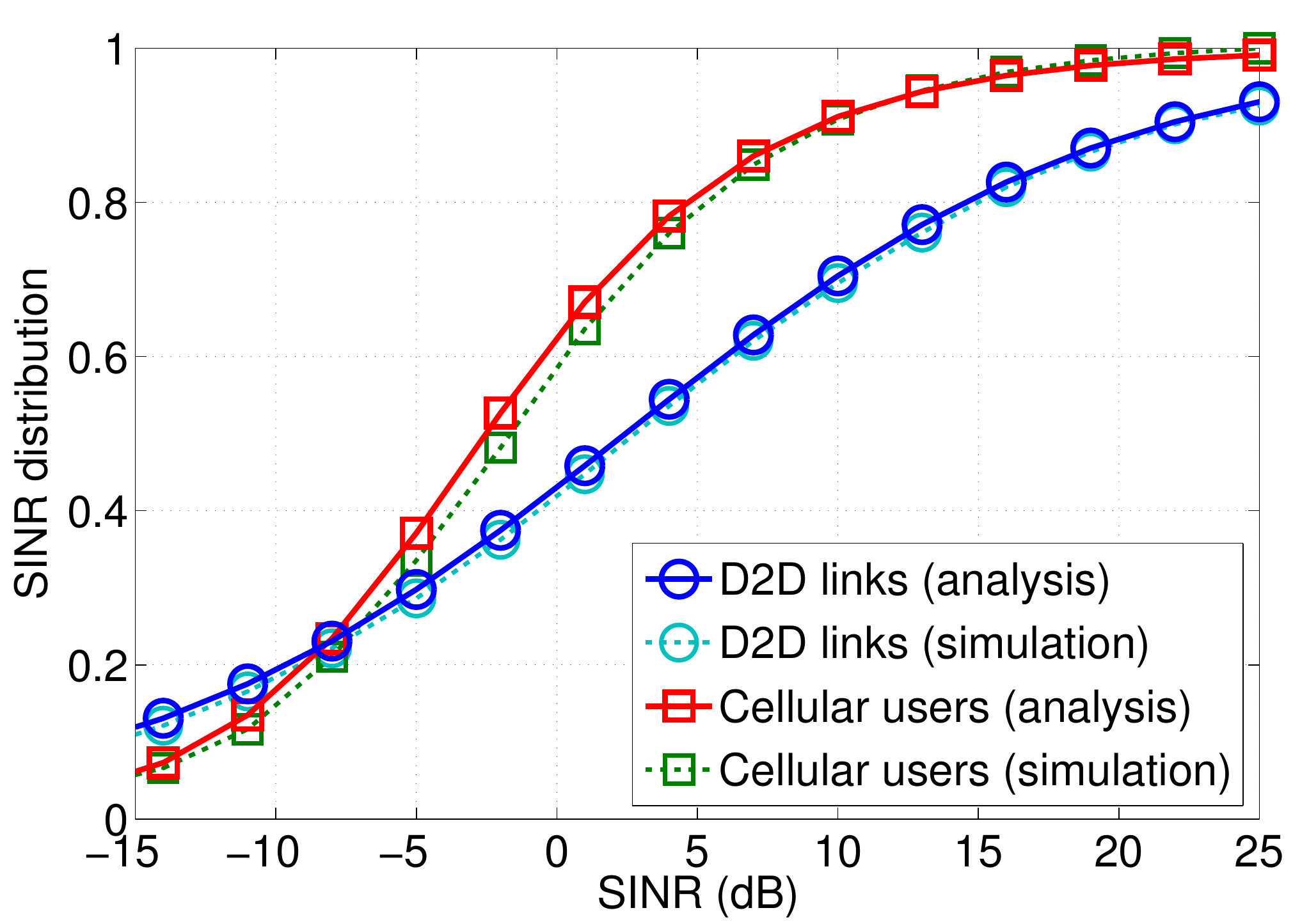}
\caption{The SINR CDFs of active D2D links and cellular users in the dedicated network, with hopping probabilities $p_{t_1}=p_{t_2}=1$, $p_{f_1}=0.2$ and $p_{f_2}=0.6$.  The theoretical and simulation results are in quite good agreement.}
\label{fig:overlaysinr}
\end{figure}

\begin{figure}
\centering
\includegraphics[width=8cm, height=6cm]{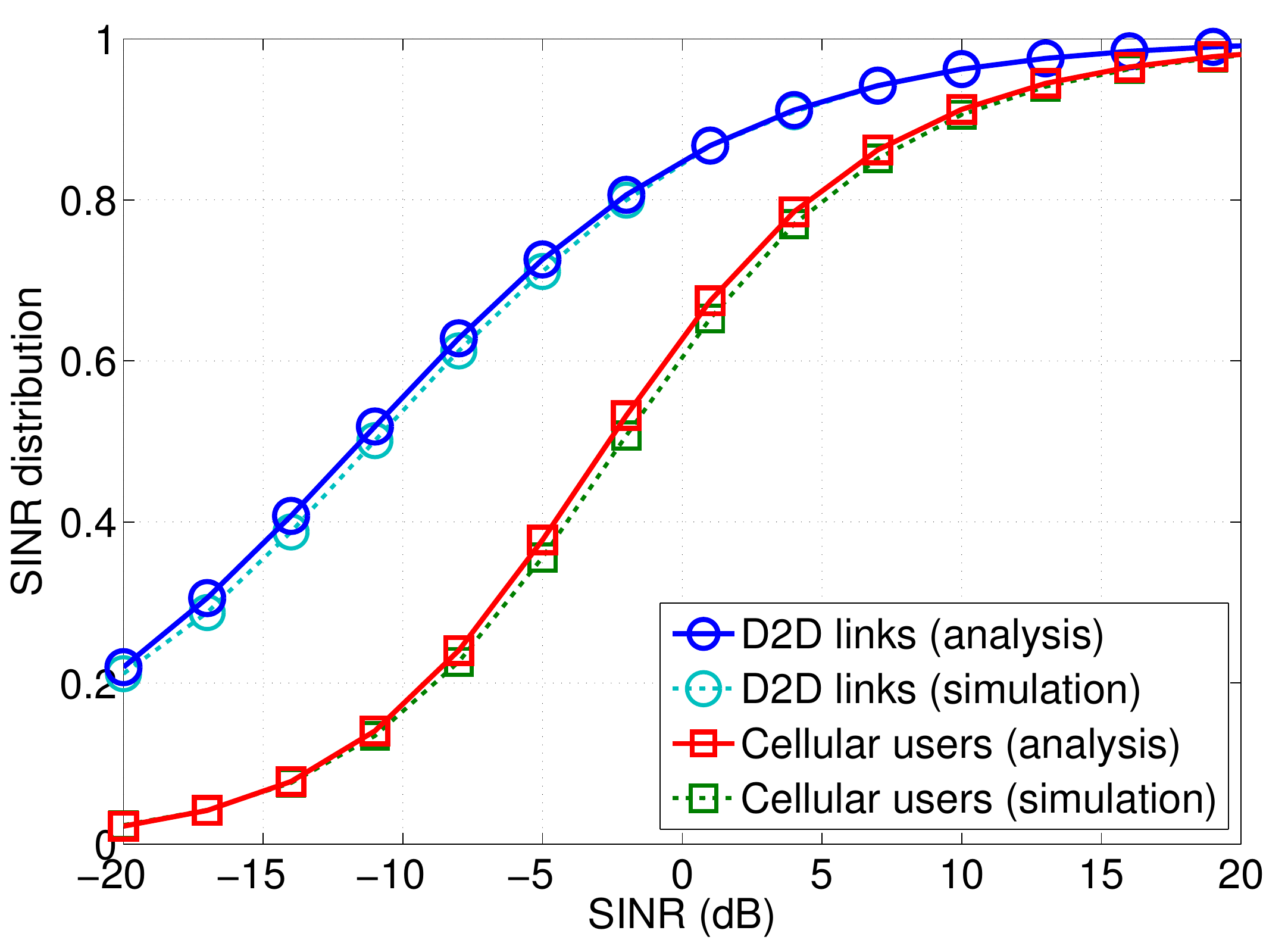}
\caption{The SINR CDFs of D2D links and cellular users in the shared network. The hopping probabilities are $p_{t_1}=p_{t_2}=1$, $p_{f_1}=0.1$ and $p_{f_2}=0.3$. The theoretical results are almost the same as the simulation result.}
\label{fig:underlaysinr}
\end{figure}

We validate the anaytical results of rates in Figs. \ref{fig:rateoverlay} and \ref{fig:rateunderlay}. We fix the ratio of D2D density to cellular user density (e.g.,  $1/2$), and increase these two densities proportionally. The analytical results are almost the same as the simulation results. 
The average rates of both cellular and potential D2D links decrease as the density increases, due to the decreasing available resources per link, as well as the increasing interference. Comparing the dedicated and shared networks, the D2D links have much higher average rate in the dedicated network. This implies that in a hybrid network sharing downlink resources, the interference from BSs may significantly degrade the network performance.
We can observe that the D2D rates in shared networks first decrease very fast, and then much more slowly when the BSs become fully loaded.  Indeed, when BSs are lightly loaded, the interference from BSs increases as the user density increases, which makes the D2D SINR decrease. In the fully loaded case, the interference from BSs stays almost the same. Though the interference from other D2D links increases, the decrease of D2D rate slows down, which implies that the interference from BSs is dominant in the performance of D2D links.
 Though the lower bound of rates are not very tight, the shapes are almost the same as the exact simulated rates, providing possibilities for optimization in terms of simple closed-form lower bounds. We compare the performance of exact rates and their lower bounds in the following subsection.

\begin{figure}
\centering
\includegraphics[width=7.5cm, height=6cm]{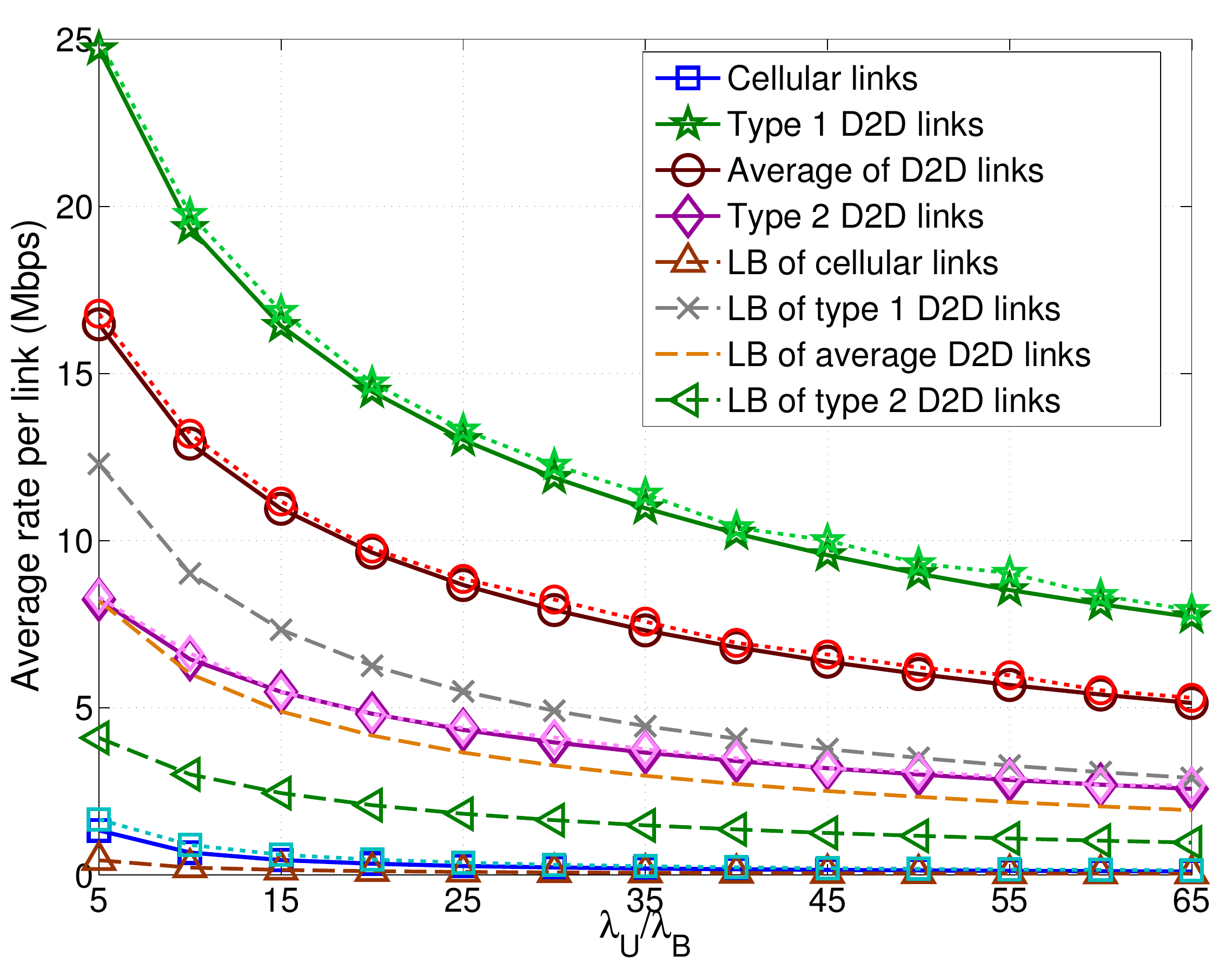}
\caption{The average rates vs. the density of users in the dedicated network ($\theta=0.5$). The hopping probabilities are $p_{t_1}=p_{t_2}=1$, $p_{f_1}=0.2$ and $p_{f_2}=0.6$. The density of potential D2D links increases proportionally to the density of cellular users. The dashed lines are the simulation results while the solid lines are the corresponding analytical results. The lower bound of rates are not very tight, but the shapes are almost the same as the exact rates.}
\label{fig:rateoverlay}
\end{figure}

\begin{figure}
\centering
\includegraphics[width=8cm, height=6cm]{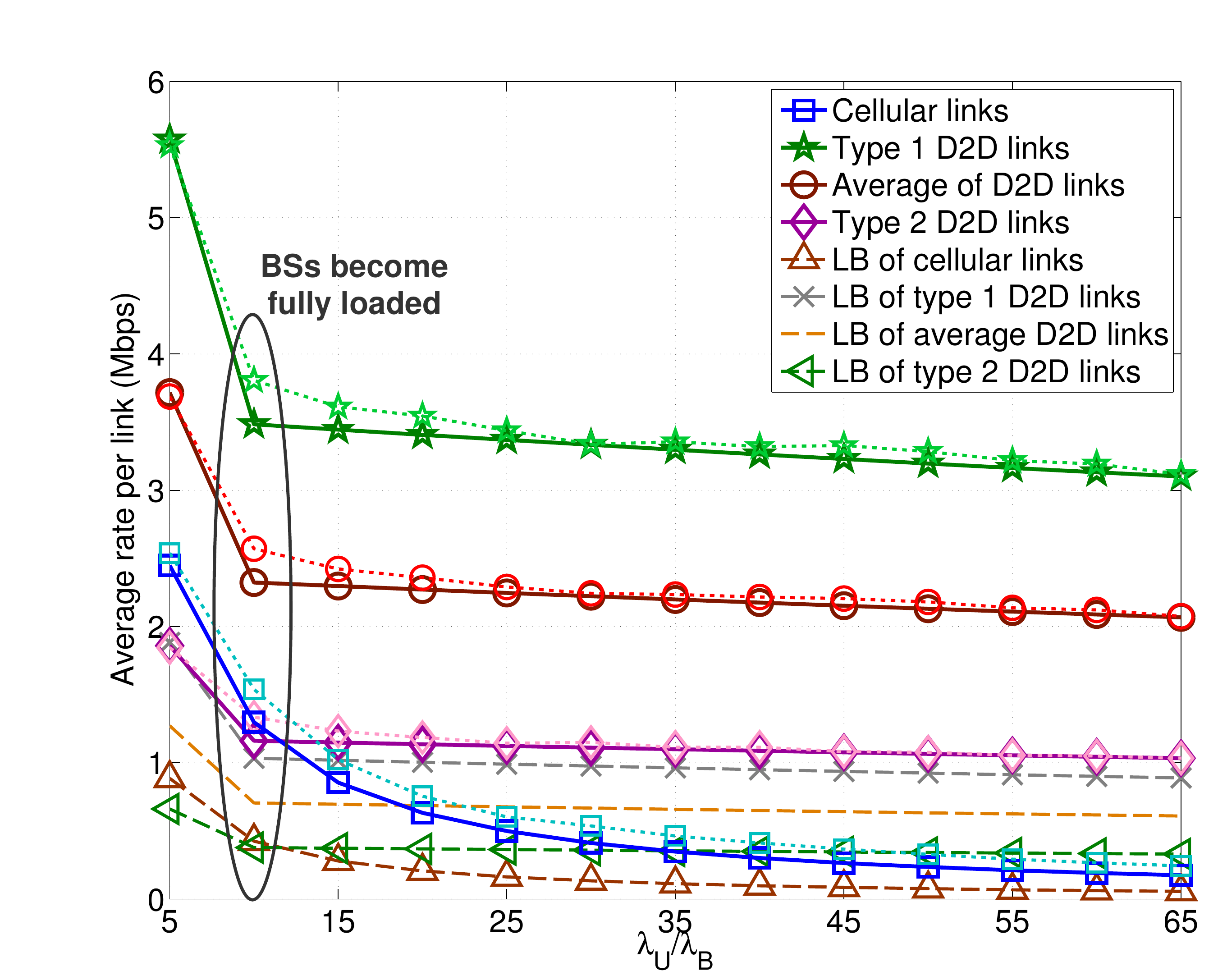}
\caption{The average rates vs. the density of users in the shared network. The hopping probabilities are $p_{t_1}=p_{t_2}=1$, $p_{f_1}=0.1$ and $p_{f_2}=0.3$. The dashed lines are the simulation results while the solid lines are the corresponding analytical results. Similar to dedicated network, the lower bound of rates are not very tight, but the shapes are almost the same as the exact rates.}
\label{fig:rateunderlay}
\end{figure}

\vspace{-0.4cm}
\subsection{Optimization of Network Performance}
The variation of rate density with time and frequency hopping probabilities in a heavily loaded network are shown in Figs. \ref{fig:drate-pt-60} and \ref{fig:drate-pf-60}, respectively. As we can observe, the optimal hopping probabilities to maximize rate lower bounds are the same as the ones to maximize the exact rates.
To maximize rate density,  the active D2D links access the frequency resource according to their service demands in both  dedicated and shared networks (i.e.,~$p_{f_i}^*=\min\{1, b_{D_i}/B_C\}$). All the potential D2D traffic is transmitted directly by D2D to alleviate the heavy load situation in the cellular network, and thus to maximize the total rate. Note that the optimal mode selection may be different for other objective functions. As shown in these two figures, the overall rate with dedicated allocation is greater than shared allocation in heavily loaded networks. One possible reason is that the rate of D2D links in the dedicated network is much larger than in the shared network, where the interference from BSs may limit the network performance, as it is observed in Figs.~\ref{fig:rateoverlay} and \ref{fig:rateunderlay}. By appropriately allocating resources between D2D  and cellular networks, the active D2D links can get a quite large rate compared to cellular UEs. However, the shared network may overwhelm the dedicated network without optimal resource partition, which is investigated in Fig. \ref{fig:drate-theta}. 


\begin{figure}
\centering
\includegraphics[width=8cm, height=6cm]{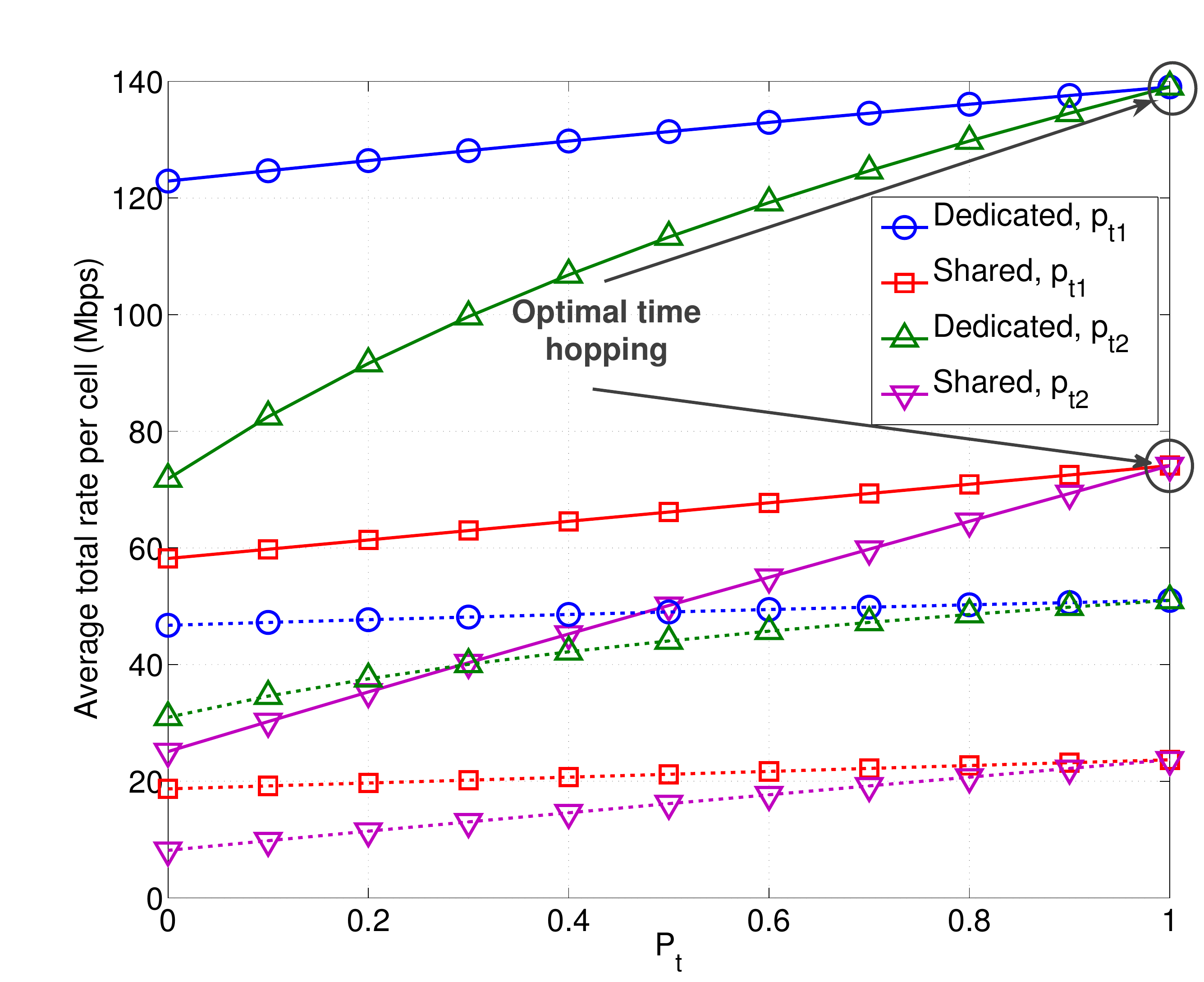}
\caption{Effect of time hopping probabilities on the total rate density in heavily loaded networks ($\theta=0.5$). The solid curves and dashed curves show the performance of exact rates and their lower bounds, respectively. We let frequency hopping probabilities be $p_{f_i}=\min\{1, B_C/b_{D_i}\}$.
The optimal time hopping probabilities to maximize the total average rate are the same as our analytical solutions (i.e. $p_{t_1}^*=p_{t_2}^*=1$). Though the rate lower bounds are not very tight, the time hopping probabilities to maximize the rate lower bounds and the exact rate are the same.}
\label{fig:drate-pt-60}
\end{figure}

\begin{figure}
\centering
\includegraphics[width=8cm, height=6cm]{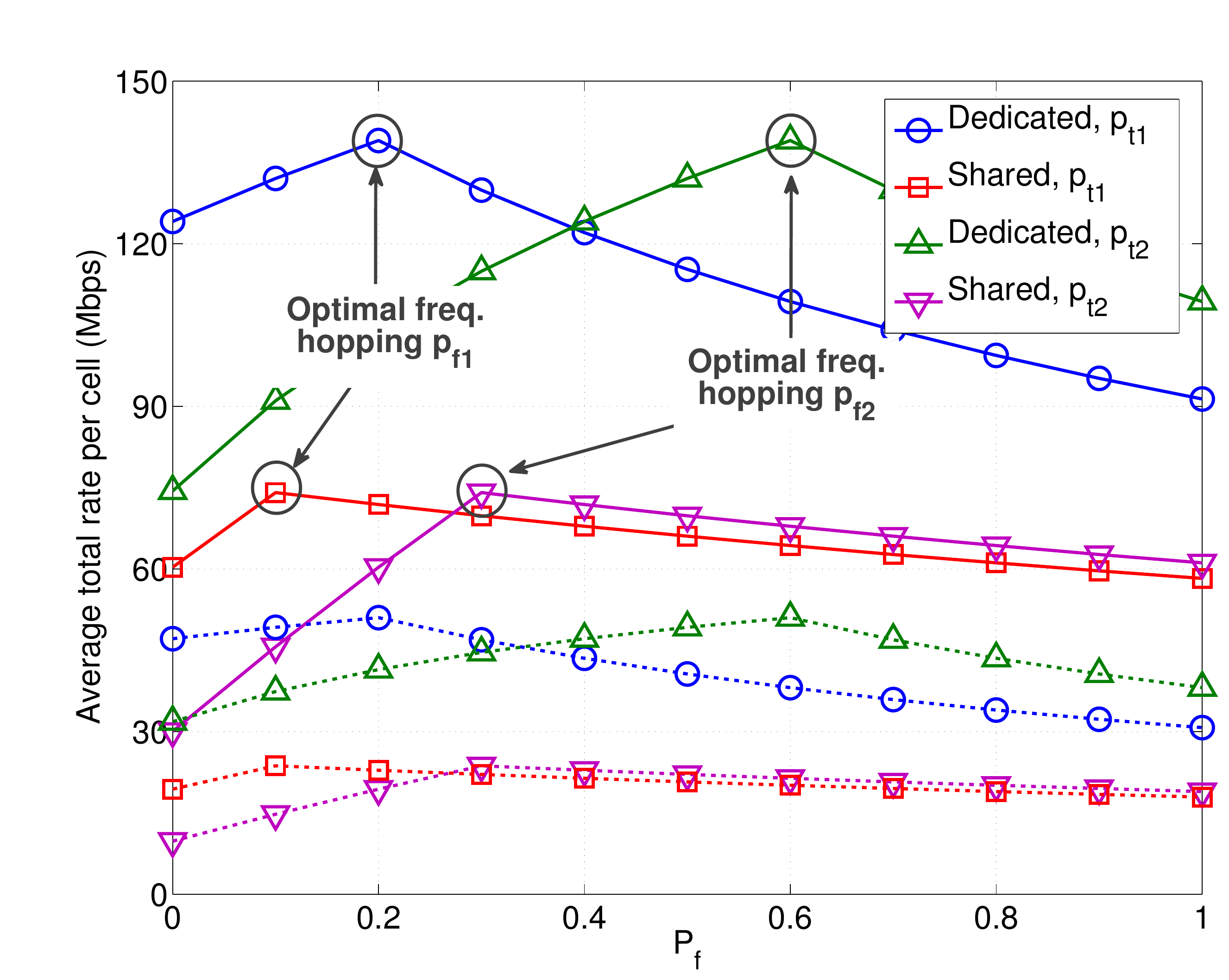}
\caption{Effect of frequency hopping probabilities  on the total rate density in heavily loaded networks ($\theta=0.5$). The solid curves and dashed curves show the performance of exact rates and their lower bounds, respectively. We let the time hopping probabilities be $p_{t_i}=1$. The optimal frequency hopping probabilities are the same as our analytical solutions (i.e. $p_{f_i}^*=\min\{1, B_C/b_{D_i}\}$).}
\label{fig:drate-pf-60}
\end{figure}

Though in the setting of this paper, the dedicated network have a greater rate  than the shared network, the conclusion differs in different scenarios. For example, when we add an additional condition for guaranteeing the cellular network performance in the dedicated network (e.g., $\theta\leq 0.1$), the optimal total rate in the shared network would be greater than the total rate in the dedicated network, which is illustrated in Fig. \ref{fig:drate-theta}. Another example is the network with a small $\lambda_D$, which may have a better performance using the shared approach (e.g., with $\lambda_D=0.1\lambda_B$, the total rates per cell of dedicated and shared networks are 9.6 and  13.6 Mbps, receptively). 
Therefore, there is no absolute advantage for the spectrum allocation approaches in general settings.

\begin{figure}
\centering
\includegraphics[width=8cm, height=6cm]{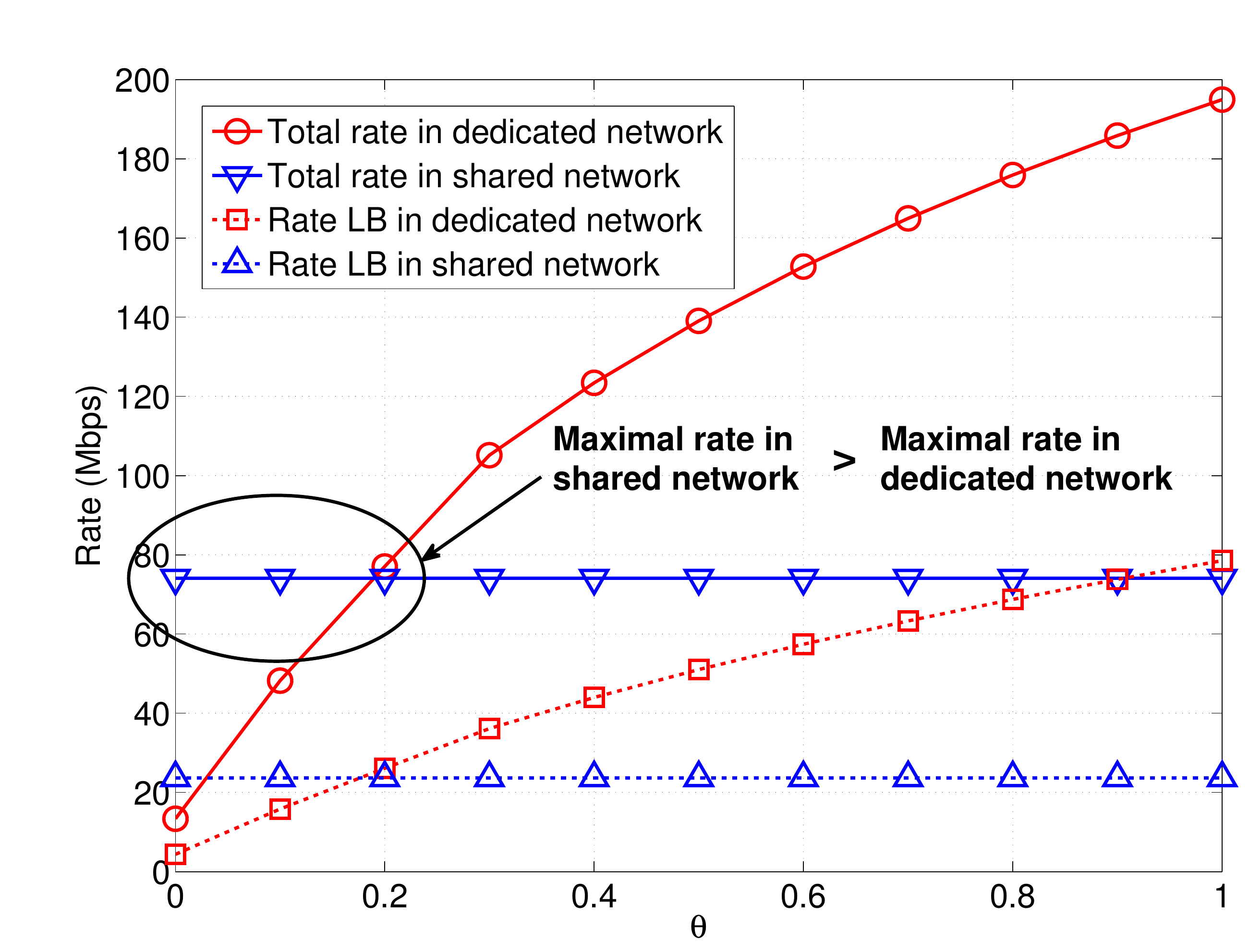}
\caption{Total rate versus $\theta$. The solid curves and dashed curves show the performance of exact rates and their lower bounds, respectively.  We let frequency hopping probabilities be $p_{f_i}=\min\{1, B_C/b_{D_i}\}$ and the time hopping probabilities be $p_{t_i}=1$. The total average rate of the shared network is greater than the dedicated network when $\theta$ is small.}
\label{fig:drate-theta}
\end{figure}

Fig. \ref{fig:drate-theta} also shows that the optimal resource partition in our simulation setup to maximize the total rate is $\theta^*=1$. We can have  different $\theta^*$ if the system parameters change. For example, Fig. \ref{fig:drate-theta-delta} shows that $\theta^*=\frac{b_{D_2}}{B}$ when the average distance between the D2D transmitter and its receiver increases to~280m. This is consistent with the conclusion made in Proposition \ref{prop:theta}, where we claim that $\theta^*$ depends on various network parameters (e.g., $\delta$) and  belongs to the set $\mathcal{O}$. Note that we get the solution $\theta^*=0$ or $1$, which is unfair,  due to that we consider total rate maximization as our objective function for the first-cut study. The optimal value of $\theta$ would be very different if other utility functions are considered. For example, for maximization of log-rate, we will never get $\theta^*=0$ or $1$ (techniques similar to \cite{lin2013uplinkD2D} can be used for this analysis).  We leave the investigation of other utility functions to future work.

\begin{figure}
\centering
\includegraphics[width=8cm, height=6cm]{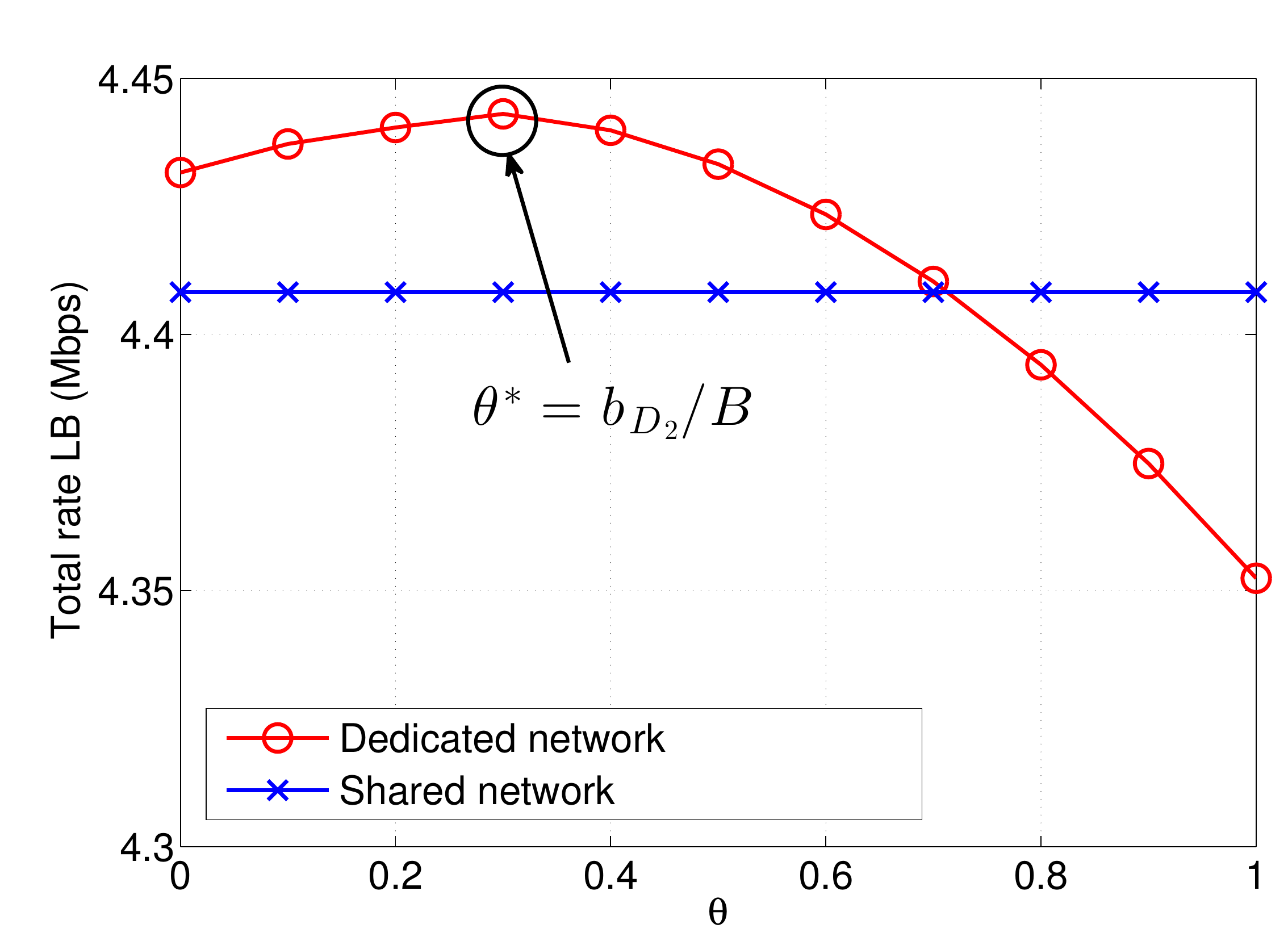}
\caption{Rate versus $\theta$ in a network with the average distance between a D2D transmitter and its receiver being 280m.  We let frequency hopping probabilities be $p_{f_i}=\min\{1, B_C/b_{D_i}\}$ and the time hopping probabilities be $p_{t_i}=1$. The optimal $\theta^*$ depends on the network parameters. The simulation result is consistent with the conclusion in Proposition \ref{prop:theta}.}
\label{fig:drate-theta-delta}
\end{figure}

\begin{figure}
\centering
\includegraphics[width=8cm, height=6cm]{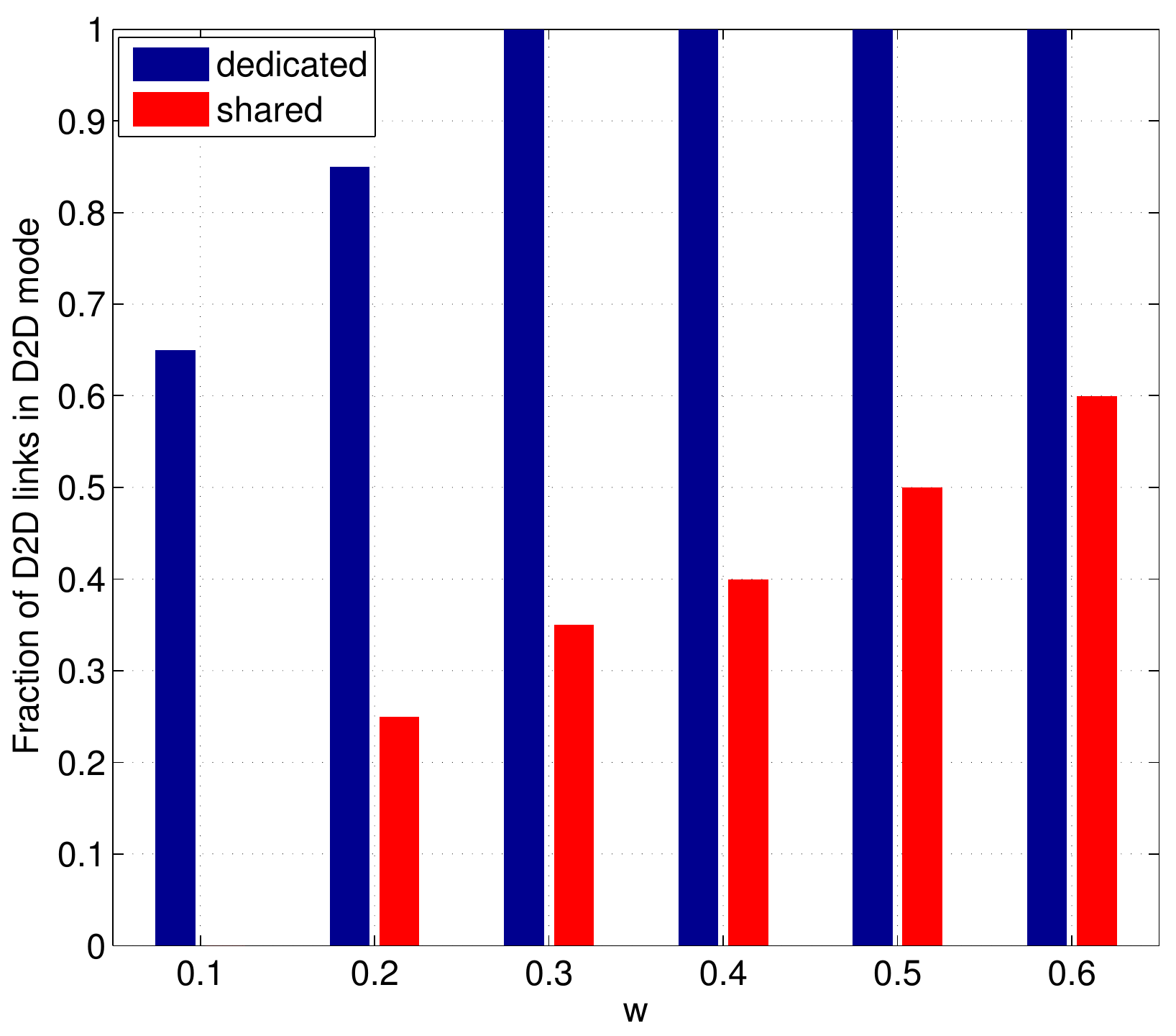}
\caption{Effect of parameter $w$ on the optimal mode selection to maximize the total rate. As $w$ increases, the number of links in D2D mode increases. All potential D2D links would be in D2D mode when $w\geq 1$.}
\label{fig:w}
\end{figure}

Though in most cases, we have $w\geq 1$, we investigate the impact of $w$ ($w>0$ for more general cases) on the mode selection in Fig. \ref{fig:w}. As  $w$ increases, which can be interpreted as the increasing price of cellular resource, the cellular communication becomes more and more unattractive for potential D2D traffic, and thus the load is shifted from cellular networks to D2D networks, in order to maximize the total rate. Therefore, it is possible to extend current framework to a system, which can dynamically control $w$ so as to adjust the load in D2D and cellular systems to  achieve other more general utilities (e.g., utilities involving fairness). We leave the analysis to future work.

\section{Conclusion}\label{sec:conclusion}
This paper has presented tractable frameworks for both dedicated and shared networks, which provide accurate expressions for important performance metrics (i.e., the coverage probability and average rate). With an appropriate resource partitioning, we observe that the dedicated network has a larger overall rate than shared network in downlink scenario. In dedicated network, the D2D links would access the frequency bands as many as needed (i.e., $p_{f_i}^*=\min\{1, b_{D_i}/(\theta B)\}$) to maximize any non-decreasing utility function. To maximize the total rate, the potential D2D links are all in D2D mode in both fully loaded dedicated and shared networks, when $w\geq 1$. 

There are numerous extensions of the proposed flexible model, like multiple antennas, power control, interference cancellation or interference alignment, more intelligent scheduling schemes and study of other utility functions. For example, one possible extension is to use the SIR-based CSMA protocol~\cite{BacLi11}. Though the set of active D2D links is no longer a homogeneous PPP, we can approximate it to a PPP with appropriate density, at little cost of accuracy. Then we can use the proposed model in this paper to analyze the network performance. Another possible extension of the proposed framework is to model BSs as other point processes, e.g., Matern hard core process (MHC), which characterizes the repulsiveness of BSs \cite{YeAndMHC13}. Analysis on the effect of $w$ on load balancing is also of interest.



\appendices
\section{Proof of Proposition \ref{theo:cdf-d2d}}\label{pf:cdf-d2d}
Conditioning on the distance between a typical transmitter and its receiver, we have
\begin{equation*}
\begin{aligned}
\mathbb{P}\left(\sinr >\beta\mid v\right)&=\mathbb{P}\left(h_0>s(I_{\tphi_D}+\sigma^2) \mid v \right)\\
&\overset{(a)}{=}\mathbb{E}_{I_{\tphi_D}}\left[ \exp\left(-s(I_{\tphi_D}+\sigma^2) \right) \right]\\
&= e^{-s\sigma^2}\mathcal{L}_{I_{\tphi_D}}(s),
\end{aligned}
\end{equation*}
where $s=\beta P_{D}^{-1}v^{\alpha}$, and $\mathcal{L}_{I_{\tphi_D}}(s)$ is the Laplace transform of random variable $I_{\tphi_D}$. The equality $(a)$ follows from  $h_0\sim \exp(1)$, and the last equality follows from the independence of noise and interference.

The Laplace transform can be further derived as follows:
\begin{equation*}
\begin{aligned}
&\mathcal{L}_{I_{\tphi_D}}(s)=\mathbb{E}\left[ \exp\left( -s\sum_{Z_i\in{\tphi_D\setminus 0}}P_{D}h_i|Z_i|^{-\alpha} \right) \right]\\
&\overset{(a)}{=} \exp\left( -2\tilde{\lambda}_{D}\int_0^\infty\int_0^\infty \left(1-e^{-sP_{D}h/u^{\alpha}}\right)F(dh)udu \right)\\
&=\exp\left( -2\pi \tilde{\lambda}_{D}\mathbb{E}_h \left[ \int_0^\infty  \left(1-e^{-sP_{D}h/u^{\alpha}}\right) udu\right]\right)\\
&\overset{(b)}{=}\exp\left( -\pi\tilde{\lambda}_{D} \mathbb{E}_h \left[ \Gamma\left( 1-\frac{2}{\alpha}\right)(shP_{D})^{\frac{2}{\alpha}} \right]\right)\\
&=\exp\left( - \tilde{\lambda}_{D} \frac{2\pi^2 /\alpha}{\sin\left(2\pi/\alpha\right)}(sP_{D})^{\frac{2}{\alpha}}  \right),
\end{aligned}
\end{equation*}
where $F(dh)$ is the law of channel fading (e.g., $F(dh)=e^{-h}dh$ in Rayleigh fading), and $\Gamma(x)=\int_0^\infty t^{x-1}e^{-t}dt$. The equality $(a)$ follows from the Slivnyak's Theorem of a PPP and the Laplace functional of a PPP \cite{Bac09,StoKen87},  $(b)$ is obtained by changing $x=\frac{shP_{D}}{r^{\alpha}}$, and the last equality follows from the Rayleigh fading assumption. Then we complete the proof by deconditioning on $v$.


\section{Proof of Corollary \ref{cor:cdf-cd-spec}}\label{pf-cor:cdf-cd-spec}
In this special case, for D2D links, we have
\begin{equation*}
\begin{aligned}
\mathbb{P}_D(\beta)=&\int_0^\infty\exp\left(-\beta P_{D}^{-1} \sigma^2 v^{\alpha}  - \tilde{\lambda_D} \frac{2\pi^2 /\alpha\beta^{\frac{2}{\alpha}}v^2}{\sin\left(2\pi/\alpha\right)} \right.\\
& \left.-  2\pi\lambda_B  H_0(\beta, \alpha) v^2-\frac{v^2}{2\delta^2}\right) \frac{v}{\delta^2} dv\\
=&\frac{1}{ 2\delta^2 \tilde{\lambda_D}\frac{2\pi^2 /\alpha}{\sin\left(2\pi/\alpha\right)}\beta^{\frac{2}{\alpha}}  +4\delta^2\pi\lambda_B  H_0(\beta, \alpha)+1},
\end{aligned}
\end{equation*}
where the last equality is obtained by letting $x=v^2$ and calculating the integral over $x$.


As for the cellular users, according to (\ref{eq:cdf-cd-cell}), we have 
\begin{equation*}
\begin{aligned}
\mathbb{P}_C(\beta)=&\int_0^\infty \exp\left(  - \tilde{\lambda_D} \frac{2\pi^2 /\alpha}{\sin\left(2\pi/\alpha\right)}(\beta \frac{P_D}{P_B})^{\frac{2}{\alpha}} r^2 \right.\\
&\left. - 2\pi\lambda_B H_1(\beta,\alpha) r^2 -\lambda_B\pi r^2  \right) 2\lambda_B\pi r dr\\
=&\frac{1 }{\frac{ \tilde{\lambda_D}}{\lambda_B} \frac{2\pi /\alpha}{\sin\left(2\pi/\alpha\right)}(\beta \frac{P_D}{P_B})^{\frac{2}{\alpha}} + 2 H_1(\beta,\alpha) +1}.
\end{aligned}
\end{equation*}

\section{Proof of Proposition \ref{prop:drate-overlay}}\label{pf-prop:drate-overlay}
Plugging $p_{f_i}^*$ to (\ref{eq:drate-overlay}), the average rate of active D2D links in (\ref{eq:drate-overlay}) is non-decreasing with respect to $p_{t_i}$. Denoting the average rate of cellular users by $g(p_{t_i})$, we have
\begin{equation*}
\begin{aligned}
g= & \frac{7(1-\theta)B\lambda_B}{9w}\frac{ \log_2(1+\beta_C)}{2 H_1(\beta_C, \alpha) + 1} \\
&\times \frac{\left(\sum_{i=1}^{M}\lambda_{D_i}b_{D_j}(1-p_{t_j})+w\lambda_Ub_C\right)}{\left(\sum_{i=1}^{M}\lambda_{D_i}b_{D_i}(1-p_{t_i})+\lambda_Ub_C\right)},
\end{aligned}
\end{equation*}
whose first derivative with respect to $p_{t_k}$ is
\begin{equation*}
\begin{aligned}
&\frac{\partial g}{\partial p_{t_k}}=\frac{7(1-\theta)B\lambda_B}{9w}\frac{ \log_2(1+\beta_C)}{2 H_1(\beta_C, \alpha) + 1}\\ 
&\times \frac{\lambda_{D_k}b_{D_k}\lambda_Ub_C(w-1)}{\left(\sum_{i=1}^{M}\lambda_{D_i}b_{D_i}(1-p_{t_i})+\lambda_Ub_C\right)^2}\geq 0,
\end{aligned}
\end{equation*}
where the last inequality follows from the assumption that $w\geq 1$ for congested networks. Therefore, the rate density is a non-decreasing function of $p_{t_i}$ and thus $p_{t_i}^*=1$.

As for a lightly loaded network, with $w$ being very small, the second term dominates the rate density, which is non-increasing with respect to $p_{t_i}$ when $w\rightarrow 0$. Therefore, $p_{t_i}^*=0$.

\section{Proof of Proposition \ref{prop:theta}}\label{pf-prop:theta}
Denote $\beta_C=\arg\max_\beta R_{Cl}^{(O)}$ and $\beta_D=\arg\max_\beta R_{Dl}^{(O)}$. Plugging  $p_{t_i}^*=1$ to (\ref{eq:drate-overlay}),  the objective function becomes
\begin{equation}\label{eq:obj-theta}
\max_\theta \ \frac{\sum_i \lambda_{D_i}p_{f_i}^*\theta B\log_2(\beta_D+1)}{1+2\delta^2 \sum_i \lambda_{D_i}p_{f_i}^* \frac{2\pi^2 /\alpha}{\sin\left(2\pi/\alpha\right)}\beta^{\frac{2}{\alpha}}} - \frac{7B \lambda_B\log_2(\beta_C+1)}{9\left(2  H_1(\beta_C, \alpha) + 1\right)}\theta.
\end{equation}
Recall that we denote $A_i=\sum_{j\in\mathcal{S}_i} \lambda_{D_j} \tilde{b}_j B\log_2(\beta_D+1)$, $C_i=2\delta^2 \sum_{j\in\mathcal{S}_i} \lambda_{D_j} \tilde{b}_j \frac{2\pi^2 /\alpha}{\sin\left(2\pi/\alpha\right)}\beta^{\frac{2}{\alpha}}$, $D= \frac{7B \lambda_B}{9}\frac{\log_2(\beta_C+1)}{2  H_1(\beta_C, \alpha) + 1}$, $E_i=\sum_{j\in\mathcal{G}_i} \lambda_{D_j} B\log_2(\beta_D+1)$, and $F_i=2\delta^2 \sum_{j\in\mathcal{G}_i}  \lambda_{D_j} \frac{2\pi^2 /\alpha}{\sin\left(2\pi/\alpha\right)}\beta^{\frac{2}{\alpha}} + 1$. 
On the $i$th region, the first term of (\ref{eq:obj-theta}) can be written as 
\begin{equation*}
\max_\theta \ \frac{E_i\theta^2 +A_i\theta}{F_i\theta+C_i}.\vspace{-0.2cm}
\end{equation*}
The objective function is thus $\frac{E_i\theta^2 +A_i\theta}{F_i\theta+C_i} - D\theta$. The first derivative of~(\ref{eq:obj-theta}) is $\frac{E_iF_i\theta^2 +2C_iE_i\theta +A_iC_i}{(C_i+F_i\theta)^2}-D$ and the second derivative is $\frac{2C_i(E_iC_i-A_iF_i)}{(C_i+F_i\theta)^3}$. Note that $C_iE_i = A_i(F_i-1)$, i.e., $C_iE_i<A_iF_i$. We consider the following two cases. 

\noindent (1) For partitions with  $E_i\geq F_iD$, the first derivative of~(\ref{eq:obj-theta}) is non-negative, and thus the objective function is non-decreasing. In this case, we have $\theta^*={\tilde{b}_{i+1}}$. Note that $E_i-F_iD$ decreases as $i$ increases. Denoting the index of the last domain that satisfies $E_i\geq F_iD$ by $k$, the objective function keeps increasing over the first $k$ partitions, and thus $\theta^*={\tilde{b}_{k+1}}$ for the first $k$ partitions. 

\noindent (2) For the partitions with  $E_i<F_iD$, we have a positive second derivative, implying that the objective function is concave. Thus, the optimal solution in the latter case is $\left[\frac{1}{F_i}\left(\sqrt{\frac{C_i(A_iF_i-E_iC_i)}{D F_i-E_i}}-C_i\right)\right]_{\tilde{b}_i}^{\min\{1,\tilde{b}_{i+1}\}}$, where $[x]_a^b$ denotes $\min\{\max\{x,a\},b\}$. 

\noindent Combining the above two cases,  the proof is complete.

\section{Proof of Proposition \ref{prop:drate-underlay}}\label{pf-prop:drate-underlay}

When $p_{f_i}B>b_{D_i}$, the objective function is non-increasing with respect to $p_{f_i}$, and thus we have $p_{f_i}^*\leq b_{D_i}/B$. Observing that $p_{f_i}$ only appears in terms of $p_{t_i}p_{f_i}$, we change variable to $x_i=p_{t_i}p_{f_i}$. The rate density maximization problem becomes 
\begin{equation}\label{eq:opt-utility-underlay-x}
\begin{aligned}
\max\limits_{p_t, x} \quad & d_{\rt}^{(S)}(x_i, p_{t_i})\\
\text{s.t. } \quad & x_i\leq b_{D_i}/B, \\
\quad & x_i \leq p_{t_i} \leq 1,\forall i\in \Phi_D.
\end{aligned}
\end{equation}
The objective function (\ref{eq:drate-underlay}) is
\begin{equation*}
\begin{aligned}
&d_{\rt}^{(S)}(x_i, p_{t_i})\\
=&\frac{\sum_{j=1}^{M}x_{j}\lambda_{D_j} B\log_2(1+\beta_D)}{   \frac{2\pi^2 /\alpha 2\delta^2\beta_D^{\frac{2}{\alpha}}}{\sin\left(2\pi/\alpha\right)}\sum_{i=1}^{M}x_{i}\lambda_{D_i}  +4\delta^2\pi\lambda_B  H_0(\beta_D, \alpha)+1} \\
&+p_a^{(S)}\frac{\left(\sum_{j=1}^{M}\frac{b_{D_j}}{w}(1-p_{t_j})\lambda_{D_j} + b_C \lambda_U \right)\log_2(1+\beta_C)}{ \frac{\sum_{i=1}^{M}x_{i}\lambda_{D_i} }{\lambda_B} \frac{2\pi /\alpha}{\sin\left(2\pi/\alpha\right)}(\beta_C \frac{P_D}{P_B})^{\frac{2}{\alpha}} + 2  H_1(\beta_C,\alpha) +1},
\end{aligned}
\end{equation*}
where the first term is independent of $p_{t_i}$, and the second term can be written as
\begin{equation}\label{eq:2nd-pt}
A\frac{\left(\sum_{j=1}^{M}b_{D_j}(1-p_{t_j})\lambda_{D_j} + wb_C \lambda_U \right)}{\left(b_C\lambda_U + \sum_{i=1}^{M} b_{D_i}\lambda_{D_i}(1-p_{t_i})\right)},
\end{equation}
where $A=\frac{\frac{ 7B\lambda_B}{9w}}{\frac{\sum_{i=1}^{M}x_{i}\lambda_{D_i} }{\lambda_B} \frac{2\pi /\alpha}{\sin\left(2\pi/\alpha\right)}(\beta_C \frac{P_D}{P_B})^{\frac{2}{\alpha}} + 2  H_1(\beta_C,\alpha) +1}>0$. The first derivative of (\ref{eq:2nd-pt}) with respect to $p_{t_i}$ is
\begin{equation*}
A \frac{b_C\lambda_Ub_{D_i}\lambda_{D_i} (w-1)}{\left(b_C\lambda_U + \sum_{i=1}^{M} b_{D_i}\lambda_{D_i}(1-p_{t_i})\right) ^2},
\end{equation*}
which is non-negative when $w\geq 1$. Therefore, given $w\geq 1$, the objective function (\ref{eq:drate-underlay}) is a non-decreasing function of $p_{t_i}$, and we have $p_{t_i}^*=1$.

In a lightly loaded network with small $w$, similar to the proof in Appendix \ref{pf-prop:drate-overlay}, we have $p_{t_i}^*=0$.


\bibliographystyle{ieeetr}
\bibliography{flashlinqbib}

\end{document}